\documentclass[12pt]{article}
\usepackage{amsmath, amsthm, amssymb, bbm}
\usepackage{graphicx, subcaption, float}
\usepackage{booktabs, adjustbox, multirow, xcolor}
\usepackage{hyperref}
\usepackage{algorithm, algorithmicx, algpseudocode}
\usepackage[round]{natbib}
\usepackage{eqparbox}
\usepackage{setspace}

\newcommand{\blind}{0}

\algdef{SE}[DOWHILE]{Do}{doWhile}{\algorithmicdo}[1]{\algorithmicwhile\ #1}%

\newtheorem{theorem}{Theorem}

\newtheorem{lemma}{Lemma}

\theoremstyle{definition}

\theoremstyle{definition}

\begin{document}

\bibliographystyle{jasa3}

\def\spacingset#1{\renewcommand{\baselinestretch}%
{#1}\small\normalsize} \spacingset{1}

\newcommand{\bm}[1]{\mathbf{#1}}
\newcommand{\bb}[1]{\boldsymbol{#1}}
\newcommand{\ind}{\perp\!\!\!\!\perp}


\if0\blind
{
  \title{\bf Factor Importance Ranking and Selection \\ using Total Indices}
  \author{
    Chaofan Huang \hspace{1mm} and \hspace{1mm} V. Roshan Joseph\thanks{Corresponding author: roshan@gatech.edu.} \vspace{3mm} \\
    H. Milton Stewart School of Industrial and Systems Engineering, \\
    Georgia Institute of Technology, Atlanta, GA, 30332}
  \maketitle
} \fi

\if1\blind
{
  \bigskip
  \bigskip
  \bigskip
  \begin{center}
    {\LARGE\bf Factor Importance Ranking and Selection \\ using Total Indices}
\end{center}
  \medskip
} \fi

\begin{abstract}
Factor importance measures the impact of each feature on output prediction accuracy. Many existing works focus on the model-based importance, but an important feature in one learning algorithm may hold little significance in another model. Hence, a factor importance measure ought to characterize the feature's predictive potential without relying on a specific prediction algorithm. Such algorithm-agnostic importance is termed as \emph{intrinsic importance} in \citet{williamson2023general}, but their estimator again requires model fitting. To bypass the modeling step, we present the equivalence between predictiveness potential and total Sobol' indices from global sensitivity analysis, and introduce a novel consistent estimator that can be directly estimated from noisy data. Integrating with forward selection and backward elimination gives rise to \texttt{FIRST}, Factor Importance Ranking and Selection using Total (Sobol') indices. Extensive simulations are provided to demonstrate the effectiveness of \texttt{FIRST} on regression and binary classification problems, and a clear advantage over the state-of-the-art methods.
\end{abstract}

\noindent%
\vfill

\spacingset{1.5} 

\section{Introduction}
\label{sec:introduction}

An important problem in statistics and machine learning is to understand which input variables affect the output variable and to rank them according to their importance based on the available input-output data. Traditionally, this is done by fitting a linear regression model and using the standardized coefficients (t-values or p-values) to assess the significance of the variables. More advanced techniques exist that decomposes the overall $R^2$ into sum of contributions from each variable \citep{gromping2007estimators}. Nevertheless, the overly simplistic linear regression model would fail to capture non-linearities and higher-order interactions, resulting in misleading estimates for the factor importance, and more importantly, the chance of missing important variables. 

With the advent of nonparameteric regression and machine learning models, e.g., generalized additive model \citep{hastie2017generalized,lin2006component}, Gaussian process \citep{rasmussen2006gaussian}, tree models \citep{breiman2001random,friedman2001greedy}, neural networks \citep{lecun2015deep}, etc., more sophisticated learning algorithms are recommended over linear regression to improve the prediction performance. Many factor importance measures are also proposed to reflect how much the predictive accuracy of a given algorithm depends on a particular variable \citep{breiman2001random,gevrey2003review,gromping2009variable,wei2015variable,lundberg2017unified,fisher2019all,murdoch2019interpretable,wojtas2020feature}, i.e., the decrease in predictiveness (e.g. $R^2$ for regression and accuracy for classification) when a variable is dropped or permuted. While such importance measures provide some interpretability about the black-box model, there is a chance that these models would fail to capture the complex input-output relationship correctly, and thus potentially leading to biased importance measures. Sometimes it could also happen that several learning algorithms fit the data almost equally well, yet each provide different importance measures, complicating the determination of the correct importance. To reduce the risk of model dependence, an ensemble of models is also investigated \citep{fisher2019all,williamson2021nonparametric,williamson2023general}, but at a higher computational cost, especially for big data. 

Given the concern and limitation of the model-based approaches, \citet{williamson2023general} argue that the factor importance should be \emph{algorithm-agnostic}. It should measure the population-level predictive potential of the factor, without being tied to any prediction algorithm. This is referred to as the \emph{intrinsic importance} in \citet{williamson2023general}, which is defined as the predictiveness gap between the \emph{optimal} prediction function using all the variables versus that fitted on the subset after removing the factors of interest. Albeit being model-agnostic, the estimator from \citet{williamson2023general} again involves fitting predictive models twice, one on all variables and the other on the carefully selected subset. They claim that with flexible enough ensemble of learning techniques, the risk of systematic bias from model misspecification can be minimized. Further improvement to detect null importance requires sample-splitting and cross-fitting, which will result in larger computational expense.  

On the other hand, factor importance is also of interest in sensitivity analysis \citep{sensitivitybook2021}, where the objective is to identify inputs of computer code that have the most substantial influence on the output and fix the ones that the output is not sensitive to. Various techniques have been proposed, including Morris screening method \citep{morris1991factorial}, variance-based (total) Sobol' indices \citep{sobol1993sensitivity}, etc. Recently, \citet{hart2018approximation} draw the equivalence between the total Sobol' index and the predictiveness potential of a factor, more specifically, the optimal approximation error introduced without using the factor of interest. This demonstrates the potential of estimating the intrinsic factor importance via total Sobol' indices. Though many efficient Monte Carlo based methods are available for the total Sobol' indices estimation, they all require having assess to the computer model. One solution is to construct a surrogate model (Gaussian process) and analytically derive the indices assuming independence \citep{oakley2004probabilistic,chen2005analytical}, but the independence assumption rarely holds in real data applications, not to mention about the challenges associated with the model fitting. 

To overcome the aforementioned limitations, we propose an efficient algorithm for estimating total (Sobol') indices directly from the randomly scattered noisy data. Our method is built upon the nearest-neighbor estimator proposed by \citet{broto2020variance}. We further adapt our estimation algorithm with forward selection and backward elimination to propose \texttt{FIRST}, Factor Importance Ranking and Selection using Total (Sobol') indices. To the best of our knowledge, \texttt{FIRST} is the first factor importance procedure that does not involve any model fitting, and hence safeguarding against the danger of model misspecification. 
     
There is also another line of research aiming to assess the significance of the variables nonparametrically \citep{szekely2007measuring,li2012feature,chatterjee2021new,azadkia2021simple,huang2022kernel}. Although the measures from these techniques can be used for factor selection, the ranking of the importance is not reliable especially when the inputs are correlated. Other approaches consider the aggregation of local importance such as derivative \citep{bai2014kernel,ye2012learning,yang2016model}, sharing similar spirit of the derivative-based sensitivity measure \citep{sobol2009}. Again this kind of measure is valid for factor selection, but when there is strong non-linearity in the model, \citet{sobol2009} point out the derivative-based measure could yield misleading conclusion about the factor importance.  

The article is organized as follows. Seciont~\ref{sec:factor_importance_and_total_sobol_indices} reviews the total Sobol' indices and presents its equivalence to the intrinsic importance measures. Section~\ref{sec:estimating_total_sobol_index_from_noisy_data} discusses a novel procedure for estimating the total Sobol' indices directly from noisy data. Section~\ref{sec:rethinking_factor_importance} proposes \texttt{FIRST} for factor importance and selection. Extensive numerical comparison on both synthetic and real data are provided in Section~\ref{sec:simulations_on_regression_problems} for regression and Section~\ref{sec:extension_to_binary_classification_problems} for binary classification. We conclude the article with some remarks in Section~\ref{sec:conclusion}.

\section{Factor Importance and Total Sobol' Indices}
\label{sec:factor_importance_and_total_sobol_indices}

This section we define the factor importance through its predictive potential, then review the total Sobol' indices \citep{sobol1993sensitivity,homma1996importance,sobol2001global} from the global sensitivity analysis literature, and last demonstrate the equivalence between the two.  

Let $\bm{X}=(X_1,\ldots,X_p)\in\mathcal{X}\subseteq\mathbb{R}^{p}$ be a $p$-dimensional random vectors with joint distribution $\mathbb{P}_{\mathcal{X}}$. Let the scalar output $Y = f^{*}(\bm{X}) \in \mathcal{Y} \subseteq \mathbb{R}$ be some deterministic function of the input $\bm{X}$, where $f^{*}\in\mathcal{F}$, the class of functions from $\mathcal{X}\to\mathcal{Y}$ that are square integrable over $\mathbb{P}_{\mathcal{X}}$. In this section, assume that the model $f^{*}$ is known and there is no random error associated with $Y$. The randomness in $Y$ is exclusively from the uncertainty in $\bm{X}$.

Let us first formally define the notations. For any non-empty $u\subset 1\!:\!p$, let $\bm{X}_{u}$ be the elements of $\bm{X}$ with index $u$, $\mathcal{X}_{u}$ be the sample space of $\bm{X}_{u}$, and $\mathcal{F}_{u}:=\{f\in\mathcal{F}:f(\bm{x})=f(\bm{x}^{\prime}) \text{ for any } \bm{x},\bm{x}^{\prime}\in\mathcal{X} \text{ such that } \bm{x}_{u}=\bm{x}^{\prime}_{u}\}$ be the class functions in $\mathcal{F}$ that depend only on the elements of the input with index in $u$. Denote $|u|$ be its cardinality, $-u:= 1\!:\!p\backslash u$ be its complement, and $\mathcal{P}(u)$ be its power set (the set of all subsets of $u$). For singleton $\{i\}$, abbreviate $\bm{X}_{\{i\}}$ and $\bm{X}_{-\{i\}}$ to $X_{i}$ and $\bm{X}_{-i}$ respectively. 

\subsection{Factor Importance}
\label{subsec:factor_importance}

The concept of quantifying the importance of an input through its predictive power is rooted in many \emph{model-based} factor importance measures \citep{darlington1968multiple,breiman2001random,lundberg2017unified,fisher2019all}. However, when the model is misspecified, the model-based measures cannot truly reflect the relevance between the inputs and the output, e.g., linear model often fails to capture the contribution from the nonlinear effect accurately. Hence, the factor importance should be measured by its population-level predictiveness potential, i.e., it should not be contingent on any specific prediction model. This idea is proposed as the \emph{intrinsic} importance in \citet{williamson2023general}.  

Let $V(f;f^{*},\mathbb{P}_{\mathcal{X}})$ be a measure of the predictiveness for any function $f\in\mathcal{F}$, e.g., surrogate model, when the true model is $f^{*}\in\mathcal{F}$ with input distribution $\mathbb{P}_{\mathcal{X}}$. The larger the $V(f;f^{*},\mathbb{P}_{\mathcal{X}})$, the more predictive $f$ is. Some popular predictiveness measures include $R^2$ for regression, accuracy for classification, etc. It is obvious that the function in $\mathcal{F}$ that maximizes $V(f;f^{*},\mathbb{P}_{\mathcal{X}})$ is the oracle model $f^{*}$ itself, i.e., $f^{*} = \arg\max_{f\in\mathcal{F}} V(f;f^{*},\mathbb{P}_{\mathcal{X}})$. Similarly, define $f^{*}_{-u}$ be the function in $\mathcal{F}_{-u}$ that maximize $V(f;f^{*},\mathbb{P}_{\mathcal{X}})$, which is the best possible predictive model without using variables $\bm{X}_{u}$. The \emph{population-level importance} $\psi_{u}$ of the factors $\bm{X}_{u}$ relative to the all input $\bm{X}$ is defined as 
\begin{align}
    \label{eq:factor_importance}
    \psi_{u} := V(f^{*};f^{*},\mathbb{P}_{\mathcal{X}}) - V(f^{*}_{-u};f^{*},\mathbb{P}_{\mathcal{X}}),
\end{align}
which quantifies the oracle approximation error introduced by excluding variables $\bm{X}_{u}$. $\psi_{u}$ can be viewed as the optimal lower bound for the approximation error to $f^{*}(\bm{X})$ for any predictive model training only on $\bm{X}_{-u}$. By construction, $\psi_{u}\geq 0$. The larger value of $\psi_{u}$ implies that $\bm{X}_{u}$ is more important in terms of its predictive potential, and $\psi_{u}=0$ if and only if the true model $f^{*}$ does not depend on $\bm{X}_{u}$ at all. For the rest of this section we will introduce the total Sobol' indices, and then show it coincides with the intrinsic importance $\psi_{u}$ \eqref{eq:factor_importance} when $R^{2}$ is the predictiveness measure. 

\subsection{Total Sobol' Indices}
\label{subsec:total_sobol_indices}

For any function $f\in\mathcal{F}$ that is square integrable over the input distribution $\mathbb{P}_{\mathcal{X}}$, \cite{sobol1993sensitivity,sobol2001global} show that it can be decomposed into 
\begin{align}
    \label{eq:decomposition}
    f(\bm{X}) = f_{0} + \sum_{k=1}^{p}\sum_{|u|=k}f_{u}(\bm{X}_{u}),
\end{align}
where 
\begin{align*}
    f_0 &= \mathbb{E}[f(\bm{X})], \\
    f_{i}(X_{i}) &= \mathbb{E}[f(\bm{X})|X_{i}] - f_{0}, \\
    f_{i,j}(X_{i},X_{j}) &= \mathbb{E}[f(\bm{X})|X_{i},X_{j}] - f_{i}(X_{i}) - f_{j}(X_{j}) - f_{0}, \\
    \vdots & \\
    f_{u}(\bm{X}_{u}) &= \mathbb{E}[f(\bm{X})|\bm{X}_{u}] - \sum_{v\in \mathcal{P}(u)\backslash u} f_{v}(\bm{X}_{v}).
\end{align*}
When the input variables are \emph{independent}, the decomposition is unique, and $f_{u}(\bm{X}_{u})$'s have zero mean and are mutually orthogonal, yielding similar decomposition for output variance, 
\begin{align*}
    \text{Var}[f(\bm{X})] = \sum_{k=1}^{p}\sum_{|u|=k}\text{Var}[f_{u}(\bm{X}_{u})],
\end{align*}
which is also known as the ANOVA (analysis of variance) decomposition \citep{efron1981jackknife}. The Sobol' index \citep{sobol1993sensitivity} of $\bm{X}_{u}$ under independent inputs is defined as 
\begin{align}
    \label{eq:sobol_index}
    S_{u} := \frac{\text{Var}[f_{u}(\bm{X}_{u})]}{\text{Var}[f(\bm{X})]}.
\end{align}
It is easy to see that $0 \leq S_{u} \leq 1$ and $\sum_{k=1}^{p}\sum_{|u|=k}S_{u}=1$. The \emph{first-order} Sobol' indices \citep{homma1996importance} is proposed for capturing the output variance that is associated with the main effect, i.e., the first-order Sobol' index of $X_{i}$ is defined as 
\begin{align}
    \label{eq:first_sobol} 
    S_{i} := \frac{\text{Var}[f_{i}(X_{i})]}{\text{Var}[f(\bm{X})]} = \frac{\text{Var}[\mathbb{E}\{f(\bm{X})|X_{i}\}]}{\text{Var}[f(\bm{X})]} = \frac{\text{Var}[f(\bm{X})] - \mathbb{E}[\text{Var}\{f(\bm{X})|X_{i}\}]}{\text{Var}[f(\bm{X})]}, 
\end{align}
where the last equality follows from the variance decomposition formula. $S_{i}$ measures the expected proportional reduction in $\text{Var}[f(\bm{X})]$ if $X_{i}$ is fixed to some constant, so intuitively the higher the $S_{i}$, the more important $X_{i}$ is. However, $S_{i}$ fails to account for any effect from the interaction between $X_{i}$ and other inputs. To overcome this limitation, \citet{homma1996importance} proposed the \emph{total} Sobol' index $S_{i}^{\text{tot}}$ that not only capture the main effect but also all the interaction effects involving $X_{i}$, 
\begin{align}
    \label{eq:total_sobol_1}
    S_{i}^{\text{tot}} := \sum_{v\in\mathcal{P}(1:p):\{i\}\cap v \neq \emptyset} S_{v} = 1 - \sum_{v\in\mathcal{P}(1:p\backslash\{i\})} S_{v},
\end{align}
and thus $S_{i} \leq S_{i}^{\text{tot}}$ and the difference reflects the output variance attributed by the interaction effects.  However, when the inputs are \emph{correlated}, the $f_{u}(\bm{X}_{u})$'s in the decomposition \eqref{eq:decomposition} are no longer orthogonal. \citet{li2010global} propose the following generalization, 
\begin{align}
    \label{eq:generalized_sobol_index}
    \tilde{S}_{u} := \frac{\text{Cov}(f_{u}(\bm{X}_{u}),f(\bm{X}))}{\text{Var}[f(\bm{X})]},
\end{align} 
which reduces to \eqref{eq:sobol_index} for independent inputs. The first-order Sobol' indices \eqref{eq:first_sobol} and the total Sobol' indices \eqref{eq:total_sobol_1} can be defined in the same way using $\tilde{S}_{u}$, and in what follows we focus on this generalized definition of the Sobol' indices. \citet{li2010global} shows that $\sum_{k=1}^{p}\sum_{|u|=k}\tilde{S}_{u}=1$, but $\tilde{S}_{u}$ could be negative depending on $\text{Cov}(f_{u}(\bm{X}_{u}),f(\bm{X}))$, and hence it is possible to see $S_{i} > S_{i}^{\text{tot}}$ in the presence of dependent inputs. This negative contribution from the interaction effect creates ambiguity when trying to interpret Sobol' indices as a measure for factor importance \citep{chastaing2015generalized,song2016shapley,owen2017shapley,hart2018approximation}. On the other hand, \cite{kucherenko2012estimation} show that equivalently the total Sobol' index of $X_{i}$ is  
\begin{align*}
    S_{i}^{\text{tot}} = \frac{\text{Var}[f(\bm{X})] - \text{Var}[\mathbb{E}\{f(\bm{X})|\bm{X}_{-i}\}]}{\text{Var}[f(\bm{X})]} = \frac{\mathbb{E}[\text{Var}\{f(\bm{X})|\bm{X}_{-i}\}]}{\text{Var}[f(\bm{X})]}  \geq 0, 
\end{align*}
showing that $S_{i}^{\text{tot}}$ captures the expected proportional residual variance if all other factors $\bm{X}_{-i}$ are fixed. Hence, the smaller the $S_{i}^{\text{tot}}$, the less important $X_{i}$ is, and $S_{i}^{\text{tot}}=0$ if and only if the model $f$ does not include $X_{i}$. Moreover, the above definition of the total Sobol' indices generalizes to any subset $u\subset 1\!:\!p$, i.e., 
\begin{align}
    \label{eq:total_sobol_2}
    S_{u}^{\text{tot}} = \frac{\mathbb{E}[\text{Var}\{f(\bm{X})|\bm{X}_{-u}\}]}{\text{Var}[f(\bm{X})]}.
\end{align}

To address the interpretation issue of the Sobol' indices in the presence of dependent inputs, there is also a growing interest on investigating other global sensitivity analysis methods, including the density-based measure \citep{borgonovo2007new,borgonovo2011moment,plischke2013global}, derivative-based measure \citep{sobol2009}, Shapley effect \citep{owen2014sobol,song2016shapley,broto2020variance}, etc. For a comprehensive review, please see \citet{borgonovo2016sensitivity}, \citet{razavi2021future}, and \citet{sensitivitybook2021}.

\subsection{Equivalence under Regression}
\label{subsec:equivalence_under_regression}

The $R^{2}$ predictiveness measure is the conventional choice for regression. It is defined as 
\begin{align*}
    V(f;f^{*},\mathcal{P}_{\mathcal{X}})=1 - \frac{\mathbb{E}[\{f(\bm{X})-f^{*}(\bm{X})\}^2]}{\text{Var}[f^{*}(\bm{X})]},
\end{align*}
which measures the proportion of variability in $f^{*}(\bm{X})$ that can be explained by $f(\bm{X})$ under the input distribution $\mathbb{P}_{\mathcal{X}}$.

\begin{theorem}
\label{th:total_sobol_approximation}
When $R^2$ is the predictiveness measure, $\psi_{u}=S_{u}^{\text{tot}}$ for $u\subset 1\!:\!p$.   
\end{theorem}
\begin{proof}
When $R^2$ is the predictiveness measure, it is obvious that $ V(f^{*};f^{*},\mathcal{P}_{\mathcal{X}}) = 1$ and  
\begin{align*}
    f^{*}_{-u}(\bm{X}_{-u}) &= \arg\max_{f\in\mathcal{F}_{-u}}V(f;f^{*},\mathcal{P}_{\mathcal{X}}) \\
    &= \arg\max_{f\in\mathcal{F}_{-u}}\left(1 - \frac{\mathbb{E}[\{f(\bm{X}_{-u})-f^{*}(\bm{X})\}^2]}{\text{Var}[f^{*}(\bm{X})]}\right) \\
    &= \arg\min_{f\in\mathcal{F}_{-u}}\mathbb{E}[\{f(\bm{X}_{-u})-f^{*}(\bm{X})\}^2] \\
    &= \mathbb{E}[f^{*}(\bm{X})|\bm{X}_{-u}]
\end{align*}
is the projection of $f^{*}(\bm{X})$ onto $\mathcal{F}_{-u}$. It follows that the factor importance of $\bm{X}_{u}$ is 
\begin{align*}
    \psi_{u} &= V(f^{*};f^{*},\mathbb{P}_{\mathcal{X}}) - V(f^{*}_{-u};f^{*},\mathbb{P}_{\mathcal{X}}) \\
    &= \frac{\mathbb{E}[\{\mathbb{E}[f^{*}(\bm{X})|\bm{X}_{-u}]-f^{*}(\bm{X})\}^2]}{\text{Var}[f^{*}(\bm{X})]} \\
    &= \frac{\mathbb{E}[\text{Var}\{f^{*}(\bm{X})|\bm{X}_{-u}\}]}{\text{Var}[f^{*}(\bm{X})]} \\
    &= S_{u}^{\text{tot}},
\end{align*}
which coincides with the total Sobol' index of $\bm{X}_{u}$. The third equality follows from the fact that $\mathbb{E}[(\mathbb{E}[Y|X]-Y)^2]=\mathbb{E}[\text{Var}(Y|X)]$ (Lemma~\ref{lm:expecected_mse} in Appendix~\ref{appendix:proof_of_lemma}).
\end{proof}

\citet{hart2018approximation} provide an alternative proof for this approximation theoretic perspective of the total Sobol' indices. With this equivalence, computing the factor importance $\psi_{u}$ reduces to the computation of the total Sobol' index $S_{u}^{\text{tot}}$. Hence, utilizing the procedure presented in the next section for estimating the total Sobol' indices directly from data, we could circumvent the model fitting step for approximating $f^{*}$ and $f^{*}_{-u}$ that are required by the existing estimator of $\psi_{u}$ proposed by \citet{williamson2023general}. This also avoid the risk of systematic bias from model misspecification. 

\section{Estimating Total Sobol' Index from Noisy Data}
\label{sec:estimating_total_sobol_index_from_noisy_data}

In the computer experiments literature, various Monte Carlo methods have been proposed for estimating the total Sobol' indices \eqref{eq:total_sobol_2} efficiently, including Pick-and-Freeze estimator \citep{homma1996importance}, Jansen-Sobol estimator \citep{jansen1999analysis,sobol2001global,saltelli2010variance,kucherenko2012estimation}, Double Monte Carlo estimator \citep{song2016shapley}, etc. These methods all require knowing about (i) the model $f^{*}$, or at the very least, the ability to evaluate it for any given input, and (ii) the input distribution and its conditional in closed-form. However, neither is available for our problem, where the goal is to assess the importance of the variables solely based on the data $\{(\bm{x}^{(n)},y^{(n)})\}_{n=1}^{N}$. One possible way to overcome these obstacles is to fit (i) a surrogate model $\hat{f}$ on $\{(\bm{x}^{(n)},y^{(n)})\}_{n=1}^{N}$ to emulate the \emph{unknown} true model $f^{*}$ and (ii) an empirical distribution estimator on $\{\bm{x}^{(n)}\}_{n=1}^{N}$ with the capability of simulating samples from $X_{i}|\bm{X}_{-i}$. As alluded to in the previous sections, we try to avoid any model fitting because (i) it can be an arduous task by itself, especially with big data, and (ii) it can lead to estimation bias if the model is misspecified. On the other hand, Monte Carlo methods require simulating synthetic inputs, but as criticized in \citet{mase2022variable}, the trustworthiness of the importance measures derived from these synthetic inputs is questionable because they can be ``unlikely, physically impossible, or even logically impossible". Hence, in this section we propose an efficient algorithm for computing the total Sobol' indices \eqref{eq:total_sobol_2} using only the collected noisy data $\{(\bm{x}^{(n)},y^{(n)})\}_{n=1}^{N}$. 

\subsection{Design-based Double Monte Carlo Estimator}
\label{subsec:double_monte_carlo_estimator}

Let us start by reviewing the Double Monte Carlo estimator proposed in \citet{song2016shapley} that lays the crucial groundwork for the total Sobol' indices estimator using only data. Recall that the total Sobol' index $S_{i}^{\text{tot}}$ of $X_{i}$ with respect to function $f^{*}$ is defined as 
\begin{align*}
    S_{i}^{\text{tot}} = \frac{\mathbb{E}[\text{Var}\{f^{*}(\bm{X})|\bm{X}_{-i}\}]}{\text{Var}[f^{*}(\bm{X})]}.
\end{align*}
We focus on the estimation for a single index, i.e., $u=\{i\}$ for some $i\in1\!:\!p$, because it is often the interest to assess the importance of individual factor rather than a group of them, but the estimation procedure generalizes for any subset $u\subset 1\!:\!p$. 

The denominator $\text{Var}[f^{*}(\bm{X})]$ can be estimated by simple Monte Carlo using samples from $\mathbb{P}_{\mathcal{X}}$. The main difficulty comes from estimating the numerator,
\begin{align}
    \label{eq:total_sobol_effect}
    T_{i}= \mathbb{E}[\text{Var}\{f^{*}(\bm{X})|\bm{X}_{-i}\}].
\end{align}
For simplicity we refer $T_{i}$ as the total Sobol' effect of $X_{i}$, the un-normalized version of the total Sobol' index. \citet{song2016shapley} proposed to use a double loop Monte Carlo: the inner loop for the conditional variance using $N_{I}$ designs and the outer loop for the expectation using $N_{O}$ designs, i.e.,  
\begin{align}
    \label{eq:double_monte_carlo}
    \hat{T}^{\text{mc}}_{i} = \frac{1}{N_{O}}\sum_{m=1}^{N_{O}}\left\{\frac{1}{N_{I}-1}\sum_{j=1}^{N_{I}}\left(f^{*}(\bm{x}_{-i}^{(m)},x_{i}^{(m,j)})-\frac{1}{N_{I}}\sum_{l=1}^{N_{I}}f^{*}(\bm{x}_{-i}^{(m)},x_{i}^{(m,l)})\right)^2\right\},
\end{align}
where $\{\bm{x}_{-i}^{(m)}\}_{m=1}^{N_{O}}$ are i.i.d. samples from the distribution of $\bm{X}_{-i}$ and $\{x_{i}^{(m,j)}\}_{j=1}^{N_{I}}$ conditionally to $\bm{x}_{-i}^{(m)}$ is i.i.d. with respect to the distribution of $X_{i}|\bm{X}_{-i}=\bm{x}_{-i}^{(m)}$. The Double Monte Carlo estimator $\hat{T}^{\text{mc}}_{i}$ \eqref{eq:double_monte_carlo} is unbiased for $T_{i}$ \eqref{eq:total_sobol_effect}, and it is suggested that $N_{I}=3$ is sufficient \citep{song2016shapley}. When $N_{I} = 2$, the Double Monte Carlo estimator \eqref{eq:double_monte_carlo} reduces to the generalized Jansen-Sobol estimator in \citet{kucherenko2012estimation},
\begin{align*}
    \hat{T}^{\text{js}}_{i} = \frac{1}{N_{O}}\sum_{m=1}^{N_{O}}\frac{1}{2}\left(f^{*}(\bm{x}_{-i}^{(m)},x_{i}^{(m,1)}) - f^{*}(\bm{x}_{-i}^{(m)},x_{i}^{(m,2)})\right)^2.
\end{align*}

\subsection{Estimation from Clean Data}
\label{subsec:estimation_from_clean_data}

Before presenting the algorithm for estimating the total Sobol' indices from the noisy data, we first discuss the estimation using clean i.i.d. samples $\{(\bm{x}^{(n)},y^{(n)}=f^{*}(\bm{x}^{(n)}))\}_{n=1}^{N}$, i.e., there is no random noise associated with the output. It is worth emphasizing that here $f^{*}$ is neither known analytically nor assessable for evaluation, so the Double Monte Carlo estimator \eqref{eq:double_monte_carlo} is not applicable. To overcome this limitation, a consistent estimator based on the nearest-neighbor method is proposed in \citet{broto2020variance}. The key is to replace the inner loop designs $\{(\bm{x}_{-i}^{(m)},x_{i}^{(m,j)})\}_{j=1}^{N_{I}}$ by the $N_{I}$ samples in $\{\bm{x}^{(n)}\}_{n=1}^{N}$ with the values of their $-i$ entries closest to $\bm{x}_{-i}^{(m)}$. More formally, let $k_{-i}^{(m)}(j)$ be the index such that $\bm{x}_{-i}^{(k_{-i}^{(m)}(j))}$ is the $j$-th closest element to $\bm{x}_{-i}^{(m)}$ in $\{\bm{x}_{-i}^{(n)}\}_{n=1}^{N}$, with ties broken by random selection. The nearest-neighbor version of the Double Monte Carlo estimator \eqref{eq:double_monte_carlo} for $T_{i}$ is 
\begin{align}
    \label{eq:nearest_neighbor}
    \hat{T}^{\text{nn}}_{i} = \frac{1}{N_{O}}\sum_{m=1}^{N_{O}}\left\{\frac{1}{N_{I}-1}\sum_{j=1}^{N_{I}}\left(y^{(k_{-i}^{(m)}(j))}-\frac{1}{N_{I}}\sum_{l=1}^{N_{I}}y^{(k_{-i}^{(m)}(l))}\right)^2\right\},
\end{align} 
where $\{\bm{x}_{-i}^{(m)}\}_{m=1}^{N_{O}}$ for the outer loop Monte Carlo can either be all the samples $\{\bm{x}_{-i}^{(n)}\}_{n=1}^{N}$ or a random subsample (with or without replacement) of $\{\bm{x}_{-i}^{(n)}\}_{n=1}^{N}$ when $N$ is large. For convenience, in what follows we refer to \eqref{eq:nearest_neighbor} as the nearest-neighbor estimator. Under some mild assumptions on the smoothness of the input distribution $\mathbb{P}_{\mathcal{X}}$, \citet{broto2020variance} proved that $\hat{T}^{\text{nn}}_{i}$ converges to $T_{i}$ in probability as $N,N_{O}\to\infty$ for any model function $f$ that is bounded. The key is to recognize that as $N\to\infty$, any nearest neighbor $\bm{x}_{-i}^{(k_{-i}^{(m)}(j))}$ converges almost surely to $\bm{x}_{-i}^{(m)}$. It follows that $X_{i}|\bm{X}_{-i}=\bm{x}_{-i}^{(k_{-i}^{(m)}(j))}$ converges weakly to $X_{i}|\bm{X}_{-i}=\bm{x}_{-i}^{(m)}$, and hence $\{\bm{x}^{(k_{-i}^{(m)}(j))}\}_{j=1}^{N_{I}}$ can be viewed as i.i.d. samples simulated from $X|\bm{X}_{-i}=\bm{x}_{-i}^{(m)}$. Moreover, \citet{broto2020variance} further show that the rate of convergence for $\hat{T}^{\text{nn}}_{i}$ \eqref{eq:nearest_neighbor} is almost $o_{p}(N^{-1/2})$ by imposing a stronger assumption that (i) the first derivative of model function $f$ exist and is continuous and (ii) the input variables $\bm{X}$ are continuous, real-valued, restricted to a compact set where the density is lower-bounded (e.g., uniform or truncated Gaussian). The consistency and the rate of convergence is independent of (i) the problem dimension $p$ and (ii) the number of samples $N_{I}$ for estimating the conditional variance. For the denominator $\text{Var}[f^{*}(\bm{X})]$ of the total Sobol' indices, one can use the consistent estimator, $\frac{1}{N-1}\sum_{n=1}^{N}(y^{(n)}-\frac{1}{N}\sum_{l=1}^{N}y^{(l)})^2$.

\subsection{Extension to Noisy Data}
\label{subsec:extension_to_noisy_data}

Suppose now that the output is contaminated by random measurement error, i.e., 
\begin{align}
    \label{eq:noisy_output}
    Y = \tilde{f}(\bm{X},\epsilon) = f^{*}(\bm{X}) + \epsilon,
\end{align}
where $\epsilon$ is the random noise with zero mean and finite second moment and is \emph{independent} of the input $\bm{X}$. We will develop a consistent estimator for both $T_{i}=\mathbb{E}[\text{Var}\{f^{*}(\bm{X})|\bm{X}_{-i}\}]$, the total Sobol' effect of $X_{i}$ with respect to the true model $f^{*}$, and $\text{Var}[f^{*}(\bm{X})]$ given the i.i.d. noisy samples $\{(\bm{x}^{(n)},y^{(n)})\}_{n=1}^{N}$ of \eqref{eq:noisy_output}, where $y^{(n)}=f^{*}(\bm{x}^{(n)}) + \epsilon^{(n)}$. 

Consider under the framework of analysis of computer experiment that the random noise $\epsilon$ in \eqref{eq:noisy_output} can be treated as another stochastic variable, similar to $\bm{X}$, supplied as an input to the deterministic function $\tilde{f}$. The variance of $\epsilon$ is equivalent to $\tilde{T}_{\epsilon}$, the total Sobol' effect of $\epsilon$ with respect to the noisy model $\tilde{f}$, i.e.,  
\begin{align*}
    \tilde{T}_{\epsilon} = \mathbb{E}[\text{Var}\{\tilde{f}(\bm{X},\epsilon)|\bm{X}\}] = \mathbb{E}[\text{Var}\{f^{*}(\bm{X})+\epsilon|\bm{X}\}] = \mathbb{E}[\text{Var}(\epsilon)] = \text{Var}[\epsilon].
\end{align*}
The nearest-neighbor estimator \eqref{eq:nearest_neighbor} for $\tilde{T}_{\epsilon}$ using noisy data $\{(\bm{x}^{(n)},y^{(n)})\}_{n=1}^{N}$ is
\begin{align}
    \label{eq:nearest_neighbor_epsilon}
    \hat{\tilde{T}}^{\text{nn}}_{\epsilon} = \frac{1}{N_{O}}\sum_{m=1}^{N_{O}}\left\{\frac{1}{N_{I}-1}\sum_{j=1}^{N_{I}}\left(y^{(k_{1:p}^{(m)}(j))}-\frac{1}{N_{I}}\sum_{l=1}^{N_{I}}y^{(k_{1:p}^{(m)}(l))}\right)^2\right\}.
\end{align}
From the discussion in Subsection~\ref{subsec:estimation_from_clean_data}, $\hat{\tilde{T}}^{\text{nn}}_{\epsilon}$ is a consistent estimator for $\tilde{T}_{\epsilon}$, and hence also a consistent estimator for $\text{Var}[\epsilon]$. Here $k_{1:p}^{(m)}(j)$ is the index for the $j$-th nearest-neighbor to $\bm{x}^{(m)}$ in $\{\bm{x}^{(n)}\}_{n=1}^{N}$, where $\{\bm{x}^{(m)}\}_{m=1}^{N_{O}}$ for the outer loop Monte Carlo can either be all the samples $\{\bm{x}^{(n)}\}_{n=1}^{N}$ or a random subsample. Similarly, for $\tilde{T}_{i,\epsilon}$, the total Sobol' effect of $X_{i}$ and $\epsilon$ jointly with respect to the noisy model $\tilde{f}$, we have
\begin{align*}
    \tilde{T}_{i,\epsilon} = \mathbb{E}[\text{Var}\{f^{*}(\bm{X})+\epsilon|\bm{X}_{-i}\}] = \mathbb{E}[\text{Var}\{f^{*}(\bm{X})|\bm{X}_{-i}\}+\text{Var}(\epsilon)] = T_{i} + \text{Var}[\epsilon],
\end{align*}
because $\bm{X}_{-i}$ and $\epsilon$ are \emph{independent}. The corresponding nearest-neighbor estimator is 
\begin{align}
    \label{eq:nearest_neighbor_i_epsilon}
    \hat{\tilde{T}}^{\text{nn}}_{i,\epsilon} = \frac{1}{N_{O}}\sum_{m=1}^{N_{O}}\left\{\frac{1}{N_{I}-1}\sum_{j=1}^{N_{I}}\left(y^{(k_{-i}^{(m)}(j))}-\frac{1}{N_{I}}\sum_{l=1}^{N_{I}}y^{(k_{-i}^{(m)}(l))}\right)^2\right\}.
\end{align}
By the Slutzky’s theorem, $\hat{T}_{i}=\hat{\tilde{T}}^{\text{nn}}_{i,\epsilon}-\hat{\tilde{T}}^{\text{nn}}_{\epsilon}$ is a consistent estimator for $T_{i}$. Following the same line of thought, $\text{Var}[Y] = \text{Var}[\tilde{f}(\bm{X},\epsilon)] = \text{Var}[f^{*}(\bm{X})] + \text{Var}[\epsilon]$, and a consistent estimator for $\text{Var}[Y]$ is 
\begin{align}
    \label{eq:consistent_yvar}
    \widehat{\text{Var}[Y]} = \frac{1}{N-1}\sum_{n=1}^{N}\left(y^{(n)}-\frac{1}{N}\sum_{l=1}^{N}y^{(l)}\right)^2.
\end{align}
Again by the Slutzky’s theorem, $\widehat{\text{Var}[f^{*}(\bm{X})]}=\widehat{\text{Var}[Y]}-\hat{\tilde{T}}^{\text{nn}}_{\epsilon}$ is a consistent estimator for $\text{Var}[f^{*}(\bm{X})]$. Numerically this estimator is not always positive given the noisy data, we clip its value to 0 in the actual implementation. Moreover, $\text{Var}[f^{*}(\bm{X})]=0$ indicates that the output variance is exclusively from the random noise, and hence $\hat{S}_{i}^{\text{tot}} = 0$ for all $i\in1\!:\!p$. 

When $\widehat{\text{Var}[f^{*}(\bm{X})]}>0$, a consistent estimator for the total Sobol' index $S_{i}^{\text{tot}}$ of $X_{i}$ is 
\begin{align*}
    \hat{S}_{i}^{\text{tot}} = \frac{\hat{T}_{i}}{\widehat{\text{Var}[f^{*}(\bm{X})]}} = \frac{\hat{\tilde{T}}^{\text{nn}}_{i,\epsilon}-\hat{\tilde{T}}^{\text{nn}}_{\epsilon}}{\widehat{\text{Var}[Y]}-\hat{\tilde{T}}^{\text{nn}}_{\epsilon}}.
\end{align*}
Numerically the numerator can again be negative, and hence we also clip its value to 0. Let us refer to this as the Noise-Adjusted Nearest-NEighbor (\texttt{NANNE}) estimator, and the steps involved in its computation are outlined in Algorithm~\ref{algo:nanne}. K-d tree \citep{bentley1975multidimensional} can be employed to accelerate the nearest neighbor search in practice. 

\begin{figure}[!t]
    \centering
    \begin{subfigure}{0.32\textwidth}
        \centering
        \includegraphics[width=\textwidth]{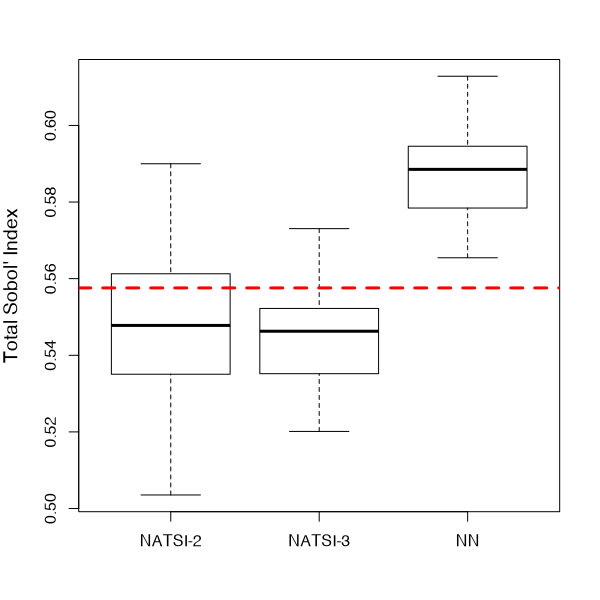}
        \caption{$X_{1}$}
    \end{subfigure}
    \begin{subfigure}{0.32\textwidth}
        \centering
        \includegraphics[width=\textwidth]{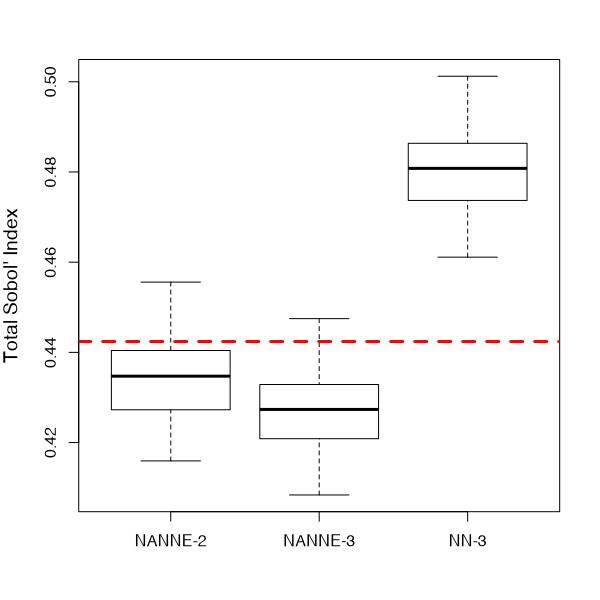}
        \caption{$X_{2}$}
    \end{subfigure}
    \begin{subfigure}{0.32\textwidth}
        \centering
        \includegraphics[width=\textwidth]{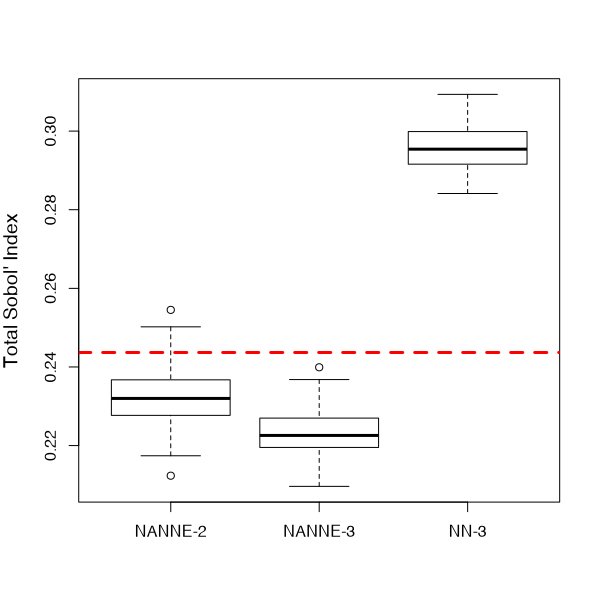}
        \caption{$X_{3}$}
    \end{subfigure}
    \caption{Comparison of \texttt{NANNE} estimator with $N_{I}=2$ (\texttt{NANNE}-2) and $N_{I}=3$ (\texttt{NANNE}-3) to the nearest-neighbor estimator \eqref{eq:nearest_neighbor} with $N_{I}=3$ (NN-3) on the noisy data simulated from the Ishigami function. The analytical total Sobol' indices are marked by the red dashed lines, and the boxplots summarize the estimations from 100 independent runs.}
    \label{fig:tsi_comp}
\end{figure}

Figure~\ref{fig:tsi_comp} compares \texttt{NANNE} and the nearest-neighbor \eqref{eq:nearest_neighbor} estimators on the noisy data $\{(\bm{x}^{(n)},y^{(n)})\}_{n=1}^{N}$ simulated from the Ishigami function \citep{ishigami1990importance}, 
\begin{align*}
    f^{*}(\bm{X}) = \sin(X_{1}) + 7\sin^2(X_{2}) + 0.1 X_{3}^{4}\sin(X_{1}), \quad X \sim\text{Uniform}[-\pi,\pi]^3.
\end{align*}
The random noise $\epsilon^{(n)}\sim\mathcal{N}(0,1)$. We run the simulation with $N=10{,}000$ samples and use all of them for the outer loop Monte Carlo designs $\{\bm{x}^{(m)}\}_{m=1}^{N_{O}}$, i.e., $N_{O}=N=10{,}000$. From Figure~\ref{fig:tsi_comp} it is evident that in the presence of noise, the nearest-neighbor estimator consistently overestimates the total Sobol' indices, whereas our proposed \texttt{NANNE} estimator is able to address this overestimation issue and provide an accurate estimate. For the nearest-neighbor estimator, \citet{broto2020variance} suggested using $N_{I}=3$ for the conditional variance estimation. However, from the simulation result in Figure~\ref{fig:tsi_comp} we find that for the \texttt{NANNE} estimator, $N_{I}=2$ yields a more robust estimation when the data is noisy. Hence, in practice for the regression problem we suggest using $N_{I}=2$ for the \texttt{NANNE} estimator. 

\begin{algorithm}[t!]
    \caption{Noise-Adjusted Nearest-NEighbor (\texttt{NANNE}) estimator.}
    \label{algo:nanne}
    \begin{algorithmic}[1]
        \State \textbf{Input:} (i) noisy samples $\{(\bm{x}^{(n)}\in\mathbb{R}^{p},y^{(n)}\in\mathbb{R})\}_{n=1}^{N}$, (ii) number of samples for inner loop Monte Carlo $N_{I}$, and (iii) number of samples for outer loop Monte Carlo $N_{O}$.
        \State Compute $\widehat{\text{Var}[Y]}$ via \eqref{eq:consistent_yvar}. 
        \State Compute $\hat{\tilde{T}}^{\text{nn}}_{\epsilon}$ via \eqref{eq:nearest_neighbor_epsilon}. 
        \State Compute $\widehat{\text{Var}[f^{*}(\bm{X})]} = \max\left(\widehat{\text{Var}[Y]}-\hat{\tilde{T}}^{\text{nn}}_{\epsilon},0\right)$.
        \If{$\widehat{\text{Var}[f^{*}(\bm{X})]}>0$}
            \For{$i=1,\ldots,p$}
                \State Compute $\hat{\tilde{T}}^{\text{nn}}_{i,\epsilon}$ via \eqref{eq:nearest_neighbor_i_epsilon}.
                \State Compute $\hat{T}_{i} = \max\left(\hat{\tilde{T}}^{\text{nn}}_{i,\epsilon}-\hat{\tilde{T}}^{\text{nn}}_{\epsilon},0\right)$.
                \State Compute $\hat{S}_{i}^{\text{tot}}=\hat{T}_{i}/\widehat{\text{Var}[f^{*}(\bm{X})]}$.
            \EndFor
        \Else
            \State Let $\hat{S}_{i}^{\text{tot}} = 0$ for $i=1,\ldots,p$.
        \EndIf
        \State \textbf{Output:} the total Sobol' indices $\{\hat{S}_{i}^{\text{tot}}\}_{i=1}^{p}$.
    \end{algorithmic}
\end{algorithm}

\section{Rethinking Factor Importance}
\label{sec:rethinking_factor_importance}

Recall from Subsection~\ref{subsec:factor_importance}, the factor importance $\psi_{u}$ of $\bm{X}_{u}$ is defined relative to all the input $\bm{X}$. The natural question that follows is \emph{what should be this full input $\bm{X}$}? For analysis of computer experiments where the model $f^{*}$ is known, the answer is straightforward, i.e. $\bm{X}$ should be all the inputs required by $f^{*}$. However, for the usual statistical modeling problem where we only have noisy data $\{(\bm{x}^{(n)},y^{(n)})\}_{n=1}^{N}$, the answer becomes ambiguous: should $\bm{X}$ be (i) all the variables that are collected or (ii) only the variables that the \emph{unknown} underlying model $f^{*}$ depends on? To answer this question, let us first consider the following linear Gaussian example,
\begin{align*}
    f^{*}(\bm{X}) = X_{1} + X_{2} + 0 X_{3},
\end{align*}
where $\bm{X}\sim\mathcal{N}(0,\bm{\Sigma})$ with $\text{Var}(X_{i})=1$ for all $i$, $\text{Cov}(X_{1},X_{2}) = \text{Cov}(X_{1},X_{3}) = 0$, and $\text{Cov}(X_{2},X_{3})=0.9$, i.e., $X_{2}$ and $X_{3}$ are strongly correlated while both are independent of $X_{1}$. Analytically solving for the factor importance relative to all input $\bm{X}=[X_{1},X_{2},X_{3}]$, we have $\psi_{1}^{(1)}=0.500$ and $\psi_{2}^{(1)}=0.095$. On the other hand, if the factor importance is computed relative to only the factors $[X_{1},X_{2}]$ that are influential in the model, we have $\psi_{1}^{(2)}=0.500$ and $\psi_{2}^{(2)}=0.500$. Clearly, the latter importance measure is evidently more in line with our expectation. The reason that $\psi_{2}^{(1)}$ is significantly smaller than $\psi_{2}^{(2)}$ is that $X_{3}$ serves as a good proxy for $X_{2}$ due to their strong correlation, though $X_{3}$ itself does not play a role in the true model. Nevertheless, if $X_{2}$ is collected as one of the predictors, we want neither $X_{3}$ show up as important nor the importance of $X_{2}$ depend on $X_{3}$. Hence, given data $\{(\bm{x}^{(n)},y^{(n)})\}_{n=1}^{N}$, \emph{factor selection precedes factor importance, and the importance should be measured relative to only factors that are relevant to the true model}. By factor selection, we imply filtering out variables that are not part of the true model, including those that could be correlated with the output but only through their dependency with the true model variables, e.g., $X_{3}$ in the linear Gaussian example. 

From the equivalence of the total Sobol' indices and the intrinsic factor importance, a factor is not important if its total Sobol' index is zero, and it should be removed from the list of inputs for the factor importance computation. Hence, a natural extension of the \texttt{NANNE} estimator to accommodate factor selection and importance simultaneously is by backward elimination. The key idea is to start with all the available factors, next compute the total Sobol' indices via \texttt{NANNE} and filter out factors with the total Sobol' index value of zero, and then re-estimate the total Sobol' indices using only the remaining variables. The above procedure is repeated until the total Sobol' indices are all positive for the surviving factors. We refer this as \texttt{NANNE-BE}, \texttt{NANNE} with Backward Elimination that provides the factor importance estimation after discarding the irrelevant variables. Please see Algorithm~\ref{algo:nanne-be} for the detailed procedure.   

\begin{algorithm}[t!]
    \caption{\texttt{NANNE-BE}: \texttt{NANNE} after Backward Elimination.}
    \label{algo:nanne-be}
    \begin{algorithmic}[1]
        \State \textbf{Input:} (i) noisy samples $\{(\bm{x}^{(n)}\in\mathbb{R}^{p},y^{(n)}\in\mathbb{R})\}_{n=1}^{N}$, (ii) number of samples for inner loop Monte Carlo $N_{I}$, and (iii) number of samples for outer loop Monte Carlo $N_{O}$.
        \State \textbf{Initialization:} $A\gets 1\!:\!p$, $\hat{S}_{i}^{\text{tot}}\gets 1/p$ for $i\in A$;
        \Do
            \State $A \gets \{i\in A: \hat{S}_{i}^{\text{tot}} > 0\}$;
            \State $\{\hat{S}_{i}^{\text{tot}}\}_{i\in A} \gets \texttt{NANNE}(\{(\bm{x}_{A}^{(n)},y^{(n)})\}_{n=1}^{N}, N_{I}, N_{O})$;
        \doWhile{$|A| > 0$ \textbf{and} $\min_{i\in A}\hat{S}_{i}^{\text{tot}} = 0$};
        \State $\hat{S}_{i}^{\text{tot}} \gets 0$ for $i \in 1\!:\!p\backslash A$.
        \State \textbf{Output:} the total Sobol' indices $\{\hat{S}_{i}^{\text{tot}}\}_{i=1}^{p}$ after factor selection. 
    \end{algorithmic}
\end{algorithm}

However, since \texttt{NANNE} estimator is developed upon the nearest-neighbor estimator, it also suffers from the \emph{curse of dimensionality}: all the points are nearly in equal distance from one another when the dimension $p$ is large. This is especially problematic when many of the factors are irrelevant for the target prediction, causing very noisy estimation and hence undermining the robustness of \texttt{NANNE} and \texttt{NANNE-BE}. To address the curse of dimensionality issue, we propose to first choose a promising set of candidates by greedy forward selection to reduce the dimensionality, and then compute the importance exclusively on the candidate factors. For the factor that is not in the candidate set, zero importance is assigned. Similar to the standard forward selection procedures \citep{efroymson1960multiple}, we start from the empty set and add one factor at a time. The selection criterion we consider is to maximize the oracle predictive power, i.e., the maximal possible variance that can be explained from the selected factors. In the context of Sobol' effect, this is the sum of the main effects and all the possible interaction effects of the selected factors. Let $u$ denote the set of indices for the selected factors, the variance of $f^{*}(\bm{X})$ that can be explained is 
\begin{align*}
    V_{u} = \sum_{v\in\mathcal{P}(u)} \text{Cov}(f^{*}_{v}(X_{v}),f^{*}(\bm{X})) = \text{Var}[\mathbb{E}\{f^{*}(\bm{X})|\bm{X}_{u}\}].
\end{align*} 
Moreover, let $\tilde{V}_{u}$ be the corresponding effect of $\bm{X}_{u}$ with respect to the noisy model \eqref{eq:noisy_output} $Y=\tilde{f}(\bm{X},\epsilon)=f^{*}(\bm{X})+\epsilon$. Since $\epsilon$ is independent of $\bm{X}$, we have
\begin{align*}
    \tilde{V}_{u} = \text{Var}[\mathbb{E}\{\tilde{f}(\bm{X},\epsilon)|\bm{X}_{u}\}] = \text{Var}[\mathbb{E}\{f^{*}(\bm{X})+\epsilon|\bm{X}_{u}\}] = \text{Var}[\mathbb{E}\{f^{*}(\bm{X})|\bm{X}_{u}\}] = V_{u}.
\end{align*}
By the variance decomposition formula, 
\begin{align*}
    V_{u} = \tilde{V}_{u} = \text{Var}[Y] - \mathbb{E}[\text{Var}\{\tilde{f}(\bm{X})|\bm{X}_{u}\}] = \text{Var}[Y] - \tilde{T}_{-u},
\end{align*}
where $\tilde{T}_{-u}$ is the total Sobol' effect \eqref{eq:total_sobol_effect} of $\bm{X}_{-u}$ with respect to the noisy model $\tilde{f}$. Given noisy data $\{(\bm{x}^{(n)},y^{(n)})\}_{n=1}^{N}$, the nearest-neighbor estimator of $\tilde{T}_{-u}$ is 
\begin{align*}
    \hat{\tilde{T}}_{-u} = \frac{1}{N_{O}}\sum_{m=1}^{N_{O}}\left\{\frac{1}{N_{I}-1}\sum_{j=1}^{N_{I}}\left(y^{(k_{u}^{(m)}(j))}-\frac{1}{N_{I}}\sum_{l=1}^{N_{I}}y^{(k_{u}^{(m)}(l))}\right)^2\right\},
\end{align*} 
where $k_{u}^{(m)}(j)$ is the index for the $j$-th nearest-neighbor to $\bm{x}_{u}^{(m)}$ in $\{\bm{x}^{(n)}\}_{n=1}^{N}$, with $\{\bm{x}_{u}^{(m)}\}_{m=1}^{N_{O}}$ be either all the samples or a random subsample of $\{\bm{x}_{u}^{(n)}\}_{n=1}^{N}$. It follows that a consistent estimator for $V_{u}$ is 
\begin{align}
    \label{eq:consistent_vare}
    \hat{V}_{u} = \widehat{\text{Var}[Y]} - \hat{\tilde{T}}_{-u},
\end{align}
where $\widehat{\text{Var}[Y]}$ is a consistent estimator of $\text{Var}[Y]$ from \eqref{eq:consistent_yvar}. The forward selection terminates when no variable can further improve the variance that can be explained. Given that the greedy forward selection could sometimes include a few irrelevant factors, we apply \texttt{NANNE-BE} (backward elimination) to filter them before computing the importance measure $\{\hat{\psi}_{i}\}_{i=1}^{p}$. The detail procedure is presented in Algorithm~\ref{algo:first}, and we name it \texttt{FIRST} which stands for Factor Importance Ranking and Selection using Total indices. Compared to the existing stepwise procedures, \texttt{FIRST} is independent of learning algorithms, i.e., all the estimation are done directly from data, and hence not only avoiding the computational burden of model fitting but more importantly the risk of model misspecification. 

\begin{algorithm}[t!]
    \caption{\texttt{FIRST}: forward selection followed by \texttt{NANNE-BE}.}
    \label{algo:first} 
    \begin{algorithmic}[1]
        \State \textbf{Input:} (i) noisy samples $\{(\bm{x}^{(n)}\in\mathbb{R}^{p},y^{(n)}\in\mathbb{R})\}_{n=1}^{N}$, (ii) number of samples for inner loop Monte Carlo $N_{I}$, and (iii) number of samples for outer loop Monte Carlo $N_{O}$.
        \State \textbf{Initialization:} $k=-1$, $A \gets \emptyset$, $\{i_{0}\}=\emptyset$, $\hat{V}_{\emptyset}=0$; 
        \Do 
            \State $k \gets k + 1$;
            \State $A \gets A \cup \{i_{k}\}$;
            \State Choose $i_{k+1}\in 1\!:\!p\backslash A$ such that $\hat{V}_{A\cup\{i_{k+1}\}}$ \eqref{eq:consistent_vare} is maximized, i.e., 
            \begin{align*}
                i_{k+1} = \arg\max_{i\in 1:p\backslash A}\hat{V}_{A\cup\{i\}};
            \end{align*}
        \doWhile{$|A| < p$ \textbf{and} $\hat{V}_{A\cup\{i_{k+1}\}} > \hat{V}_{A}$};
        \State $\{\hat{S}_{i}^{\text{tot}}\}_{i\in A} \gets \texttt{NANNE-BE}(\{(\bm{x}_{A}^{(n)},y^{(n)})\}_{n=1}^{N},N_{I},N_{O})$; 
        \State $\hat{\psi}_{i} \gets \hat{S}_{i}^{\text{tot}}$ for $i\in A$ and $\hat{\psi}_{i} \gets 0$ for $i\in 1\!:\!p\backslash A$;
        \State \textbf{Output:} the factor importance $\{\hat{\psi}_{i}\}_{i=1}^{p}$.
    \end{algorithmic}
\end{algorithm}

\section{Simulations on Regression Problems}
\label{sec:simulations_on_regression_problems}

We now compare the performance of our proposed \texttt{FIRST}\footnote{The implementation of \texttt{FIRST} is available in the R package \href{https://cran.r-project.org/web/packages/first/index.html}{\texttt{first}}.} algorithm (Algorithm~\ref{algo:first}) to a comprehensive list of factor importance/selection procedures in various regression settings. The following models are considered:
\begin{itemize}
    \item Rescaled version of the Ishigami function \citep{ishigami1990importance}:
    \begin{align}
        \label{eq:ishigami}
        f^{*}(\bm{X}) = \sin(2\pi X_{1} -\pi) + 7\sin^2(2\pi X_{2} -\pi) + 0.1 (2\pi X_{3} - \pi)^{4}\sin(2\pi X_{1} -\pi).
    \end{align}
    \item Modified version of the heavy-tailed nonlinear function \citep{huang2022kernel}:
    \begin{align}
        \label{eq:heavy-tailed}
        f^{*}(\bm{X}) = \frac{2\log(X_{1}^2+X_{2}^{4})}{\cos(X_{1})+\sin(X_{3})} + \frac{X_{2}^2\exp(X_{3})}{\sqrt{1.1-X_{6}}}.
    \end{align}
    \item Modified version of the Friedman function \citep{friedman1991multivariate}:
    \begin{align}
        \label{eq:friedman}
        f^{*}(\bm{X}) = 10\sin(\pi X_{1}X_{7}) + 20(X_{8}-0.5)^2 + 10X_{9} + 5X_{10} - 20X_{9}X_{10} - 10.
    \end{align}
\end{itemize}

\paragraph{Data Generation Mechanism} To account for the correlated input scenario, we use Gaussian copula \citep{nelsen2006introduction,kucherenko2012estimation} to model this dependency, 
\begin{align*}
    \bm{X} \sim C(G_{1}(X_{1}),\cdots,G_{p}(X_{p}); \bm{\Sigma}_{\bm{X}}) = \Phi_{p}(\Phi^{-1}(G_{1}(X_{1})),\cdots,\Phi^{-1}(G_{p}(X_{p})); \bm{\Sigma}),
\end{align*}
where $G_{i}$'s are the univariate marginal cumulative distributions, $C(\cdot;\bm{\Sigma}_{\bm{X}})$ is the Gaussian copula with correlation matrix $\bm{\Sigma}_{\bm{X}}$, $\Phi_{p}$ is the $p$-variate cumulative Gaussian distribution function with covariance matrix $\bm{\Sigma}$, and $\Phi^{-1}$ is the inverse univariate cumulative Gaussian distribution function. Note that there is a mapping from $\bm{\Sigma}_{\mathcal{X}}$ to $\bm{\Sigma}$, but for our simulation study purpose, we focus on modeling $\bm{\Sigma}$ by letting $\Sigma_{ij}=\rho^{|i-j|}$ for some $0\leq \rho \leq 1$ to control the level of correlation between the inputs. To sample $\bm{X}$, we first sample $\bm{Z}\sim\mathcal{N}(0,\bm{\Sigma})$ then apply transformation $X_{i}=G^{-1}_{i}(\Phi(Z_{i}))$ for each dimension. Uniform distribution over $[0,1]$ is used for all the marginal $G_{i}$'s, e.g., $\bm{X}$ follows the uniform distribution over the unit hypercube $[0,1]^{p}$ when $\bm{\Sigma}=\bm{I}_{p}$. To obtain the noisy output, we add a random Gaussian noise $\epsilon$ that is independent of $\bm{X}$, i.e., 
\begin{align*}
    Y = f^{*}(\bm{X}) + \epsilon, \quad \epsilon\sim\mathcal{N}(0,1).
\end{align*}
For all simulations in this section, we generate $N=1,000$ noisy samples. 

\paragraph{Groundtruth Factor Importance} Given its equivalence to the total Sobol' indices, the groundtruth factor importance measures $\{\psi_{i}^{*}\}_{i=1}^{p}$ can be computed using the Double Monte Carlo estimator \eqref{eq:double_monte_carlo} with $N_{I}=2$ and $N_{O}=100{,}000$. The Double Monte Carlo estimator requires sampling from the conditional distribution $X_{i}|\bm{X}_{-i}=\bm{x}_{-i}$, which can be done by first computing $z_{j}=\Phi^{-1}(G_{j}(x_{j}))$ for $j\in -i$, then sampling $Z_{i}$ conditional on $\bm{Z}_{-i}=\bm{z}_{-i}$, which is a conditional Gaussian distribution, and last applying transformation $X_{i}=G^{-1}_{i}(\Phi(Z_{i}))$. Recall from the discussion in Section~\ref{sec:rethinking_factor_importance}, the groundtruth factor importance is measured relative only to those variables that are presented in the true model. 

\subsection{Low Dimensional Setting}
\label{subsec:regression_low_dimensional_setting}

For factor importance we compare \texttt{FIRST} (Algorithm~\ref{algo:first}) to the following competitors:
\begin{itemize}
    \item \texttt{VIMP-R2}: sample-splitted, cross-fitted estimator \citep{williamson2023general} of the $R^{2}$-based intrinsic factor importance $\psi_{i}$ \eqref{eq:factor_importance} implemented in the R package \texttt{vimp} \citep{williamson2023vimp}. For estimating the oracle functions $f^{*}$ and $f^{*}_{-i}$ in \eqref{eq:factor_importance}, we use the random forest implemented in the R package \texttt{ranger} \citep{wright2017ranger}. For all other configurations, we follow the default values recommended in the package.
    \item \texttt{FILTER}: ``filterVarImp" implemented in the R package \texttt{caret} \citep{kuhn2020package}, which assess the importance of each factor individually by the $R^2$ of any univariate nonparametric model such as loess smoother fitted between the output and the factor.
    \item \texttt{LASSO}: Lasso \citep{tibshirani1996regression} implemented in the R package \texttt{glmnet} \citep{friedman2010glmnet}. The regularization parameters are chosen by ``lambda.1se". The factor importance is evaluated by the square of the coefficients. 
    \item \texttt{COSSO}: COSSO \citep{lin2006component} implemented in the R package \texttt{cosso} \citep{zhang2023cosso}. The regularization parameters are chosen by the default tuning procedure provided in the package. The factor importance is evaluated by the variance of each component, i.e., suppose $\hat{f}(\bm{X})=\sum_{i=1}^{p}\hat{f}_{i}(X_{i})$ is the additive model learned by COSSO, the importance of $X_{i}$ is characterized by $\text{Var}[\hat{f}_{i}(X_{i})]$.  
    \item \texttt{RF}: random forest \citep{breiman2001random} implemented in the R package \texttt{ranger} \citep{wright2017ranger}. The factor importance is assessed by the mean decrease in the impurity, which is mean squared error in the regression setting. 
    \item \texttt{RF-SHAP}: SHAP (SHapley Additive exPlanations) factor importance \citep{lundberg2017unified} based on the random forest. The Shapley value of each data instance is computed using the R package \texttt{fastshap} \citep{greenwell2021fastshap}. The factor importance is evaluated by the mean absolute Shapley value across the data \citep{molnar2020interpretable}.
    \item \texttt{RF-VS}: random forest with variable selection \citep{genuer2010variable} implemented in the R package \texttt{VSURF} \citep{genuer2015vsurf}. Random forest are fitted only on the selected variables obtained from the "interpretation step". The importance is defined by the mean decrease in the impurity for the selected variables, and 0 otherwise.  
\end{itemize}

\begin{figure}[t!]
    \centering
    \includegraphics[width=\textwidth]{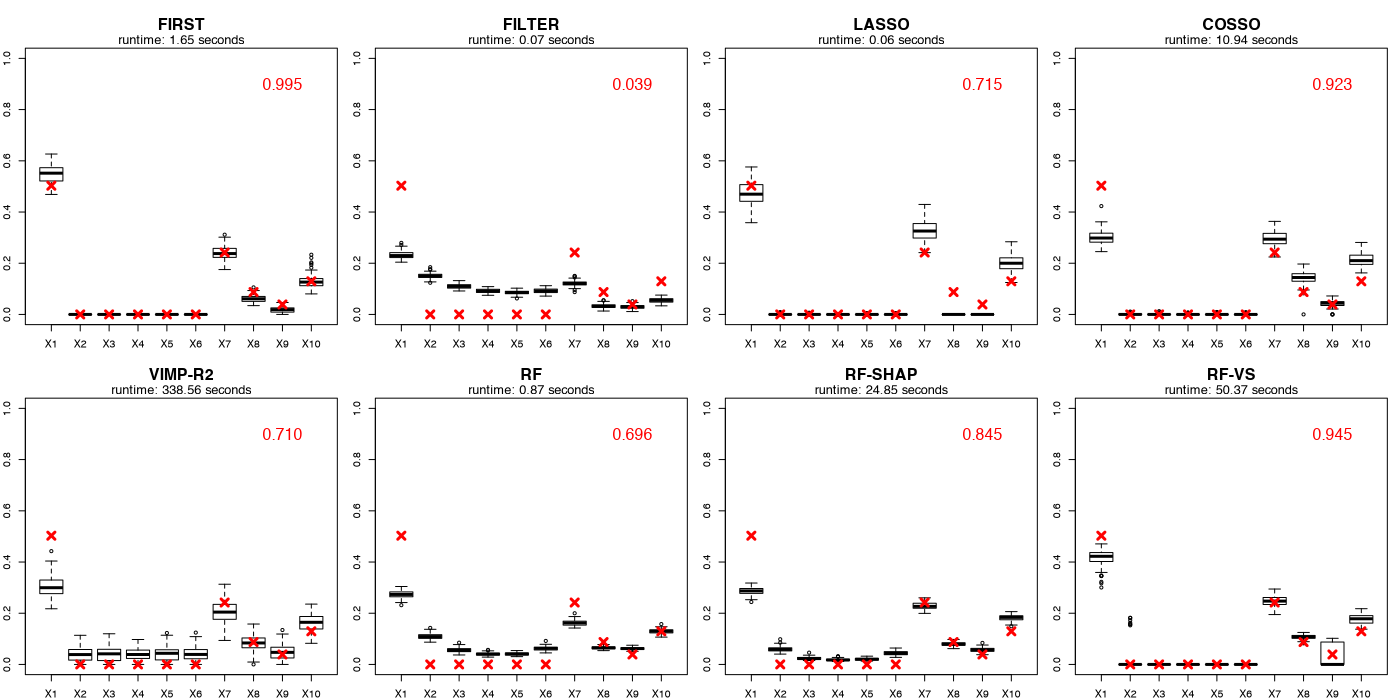}
    \caption{Comparison of \texttt{FIRST} to other factor importance procedures on the Friedman function \eqref{eq:friedman} with dimension $p=10$ and correlation $\rho=0.8$. For a better visual comparison, the importance measures are normalized to have a sum of one. For each procedure, 100 independent runs are performed, and the results are summarized by the boxplots, with the groundtruth marked by the red crosses. The average Kendall rank correlation coefficients \eqref{eq:kendall_tau}, the higher the better, are shown in the top right corner.}
    \label{fig:friedman_08}
\end{figure}

For \texttt{FIRST}, we utilize $N_{I}=2$ nearest-neighbor for estimating the conditional variance. Figure~\ref{fig:friedman_08} compares their performance on the $p=10$ dimensional Friedman function \eqref{eq:friedman} with input correlation $\rho=0.8$. Our proposed \texttt{FIRST} procedure visually aligns the best with the groundtruth importance (red crosses), while some competitors, e.g. \texttt{FILTER}, \texttt{RF}, and \texttt{RF-SHAP}, struggle from the correlation in the input. The simple linear model \texttt{LASSO} overlooks the subtle nonlinear effect. \texttt{COSSO} and \texttt{RF-VS} are able to provide reasonable estimations, but at the computational costs that are about 6 and 30 times that of \texttt{FIRST}, respectively. Similar to \texttt{FIRST}, \texttt{VIMP-R2} also aims to directly estimate intrinsic factor importance $\psi_{i}$ \eqref{eq:factor_importance} but through building estimators for the oracle functions $f^{*}$ and $f_{-i}^{*}$. From Figure~\ref{fig:friedman_08} we can see that \texttt{VIMP-R2} (i) is much more computationally expensive and (ii) has difficulty in filtering out variables that are deemed unimportant, but it does provide good importance estimation for the last four variables. To quantify the performance comparison, we consider the Kendall rank correlation coefficient $\tau$ between the groundtruth $\{\psi_{i}\}_{i=1}^{p}$ and the estimation $\{\hat{\psi}_{i}\}_{i=1}^{p}$, 
\begin{align}
    \label{eq:kendall_tau}
    \tau = \frac{2}{p(p-1)}\sum_{i < j}\text{sign}(\psi_{i}-\psi_{j})\times \text{sign}(\hat{\psi}_{i}-\hat{\psi}_{j}),
\end{align}
with value of 1 indicated perfect agreement between the two rankings, and -1 for perfect disagreement. This metric is chosen because for factor importance, the emphasis is on the ranking rather than the exact value. This also offers a fair comparison among different methods, since some captures the importance in squared error while other captures in absolute error. Again from Figure~\ref{fig:friedman_08}, \texttt{FIRST} stands out as the clear winner that best correlates in ranking with the groundtruth. 

\begin{figure}[!t]
    \centering
    \begin{subfigure}{0.32\textwidth}
        \centering
        \includegraphics[width=\textwidth]{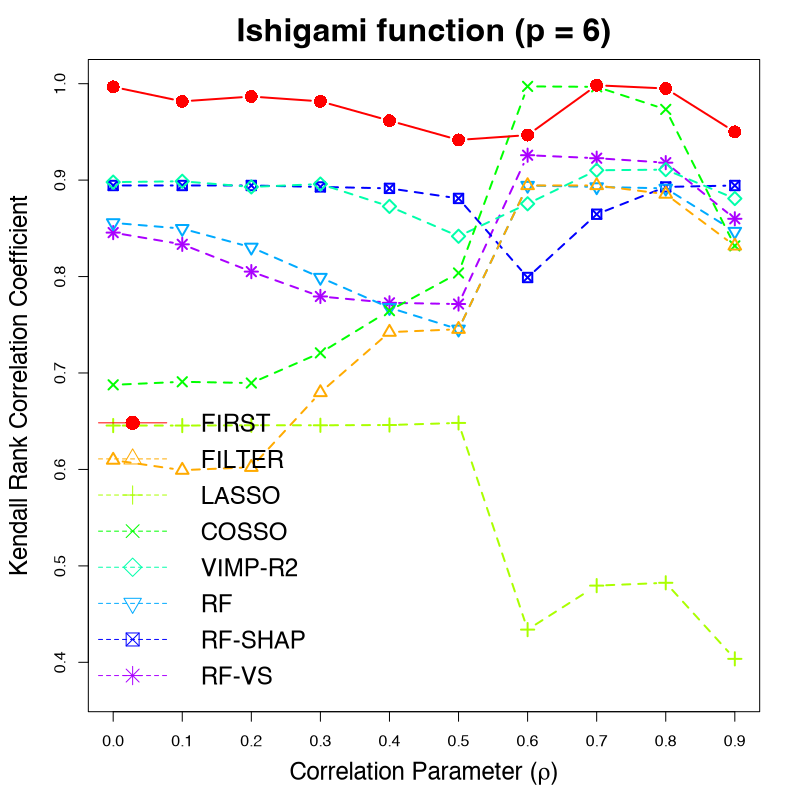}
        \caption{Ishigami \eqref{eq:ishigami}}
    \end{subfigure}
    \begin{subfigure}{0.32\textwidth}
        \centering
        \includegraphics[width=\textwidth]{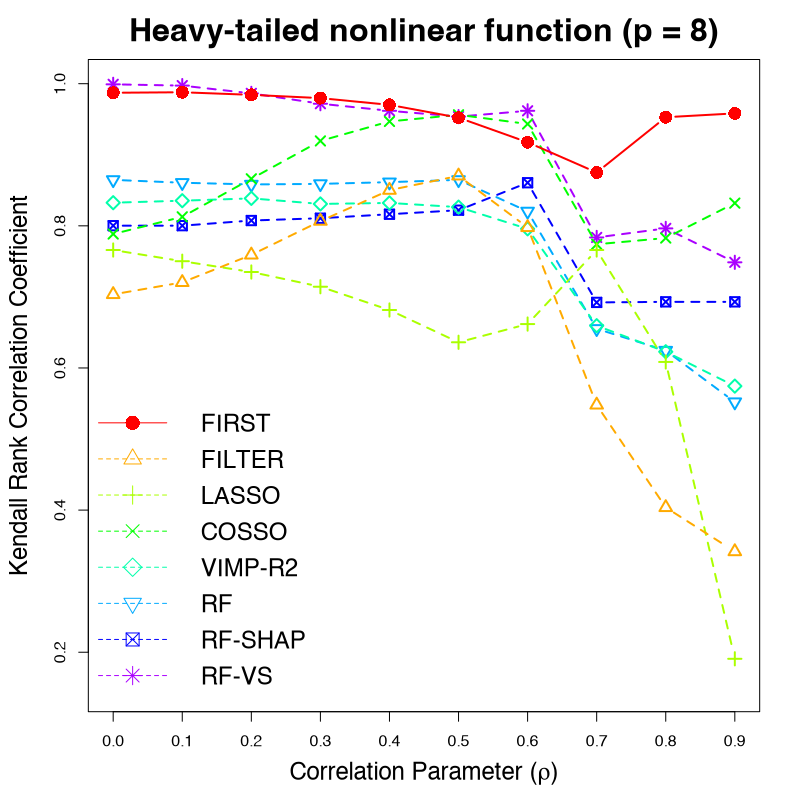}
        \caption{Heavy-tailed \eqref{eq:heavy-tailed}}
    \end{subfigure}
    \begin{subfigure}{0.32\textwidth}
        \centering
        \includegraphics[width=\textwidth]{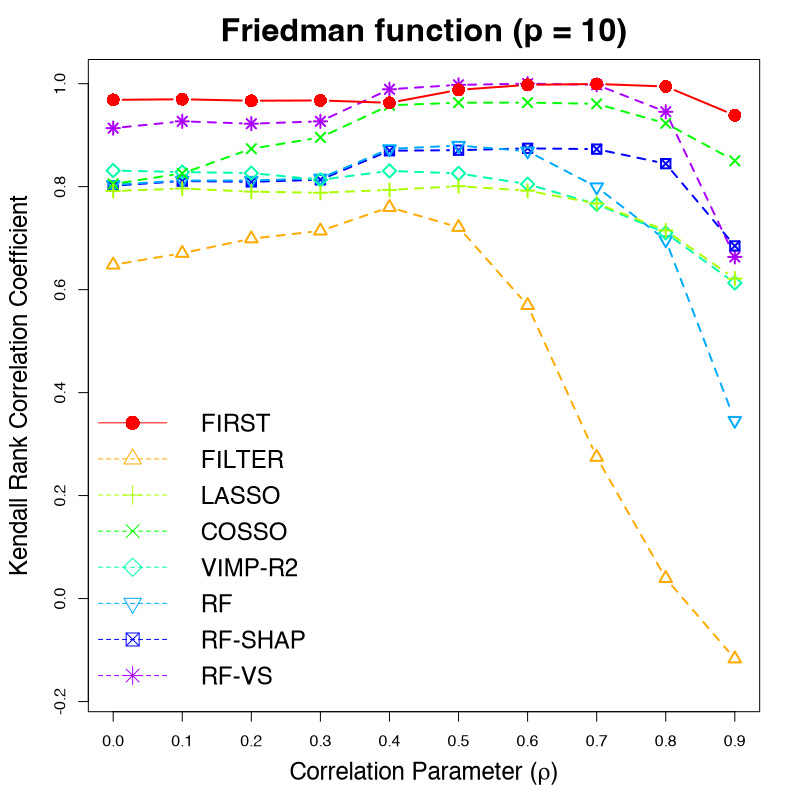}
        \caption{Friedman \eqref{eq:friedman}}
    \end{subfigure}
    \caption{Comparison of \texttt{FIRST} to other factor importance procedures on the Ishigami, Heavy-tailed, and Friedman function across different correlation ($\rho$) levels. The Kendall rank correlation coefficients \eqref{eq:kendall_tau} is averaged over 100 independent runs.}
    \label{fig:regression_importance}
\end{figure}

A more comprehensive comparison is provided in Figure~\ref{fig:regression_importance} for different correlation level on (i) Ishigami function \eqref{eq:ishigami} with $p=6$, (ii) heavy-tailed nonlinear function \eqref{eq:heavy-tailed} with $p=8$, and (iii) Friedman function \eqref{eq:friedman} with $p=10$. We can see that most of the time \texttt{FIRST} outperforms the other methods, and it is robust to the various input dependency level. \texttt{COSSO} and \texttt{RF-VS} occasionally take the lead, but their performances vary across different correlation level, especially on the Ishigami function. The other methods are evidently inferior and generally suffer from strong correlation presented in the input.  

\begin{figure}[!t]
    \centering
    \begin{subfigure}{0.32\textwidth}
        \centering
        \includegraphics[width=\textwidth]{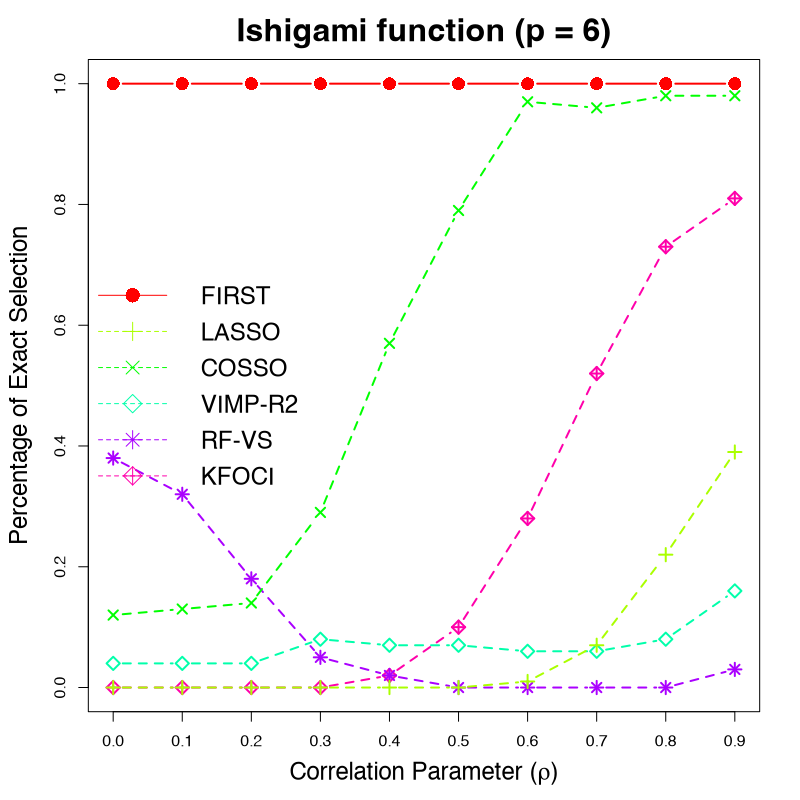}
        \caption{Ishigami \eqref{eq:ishigami}}
    \end{subfigure}
    \begin{subfigure}{0.32\textwidth}
        \centering
        \includegraphics[width=\textwidth]{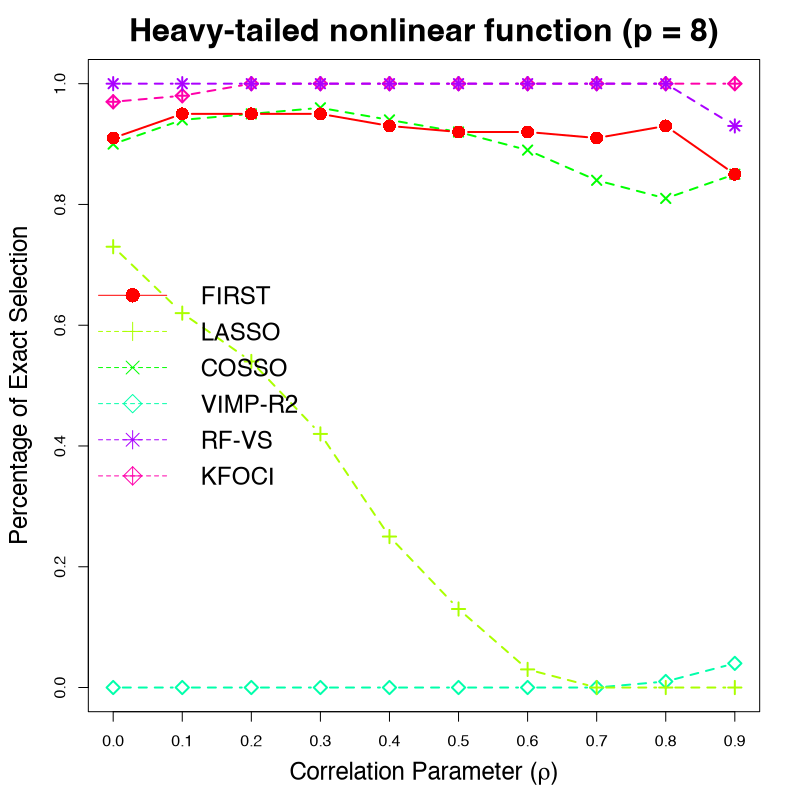}
        \caption{Heavy-tailed \eqref{eq:heavy-tailed}}
    \end{subfigure}
    \begin{subfigure}{0.32\textwidth}
        \centering
        \includegraphics[width=\textwidth]{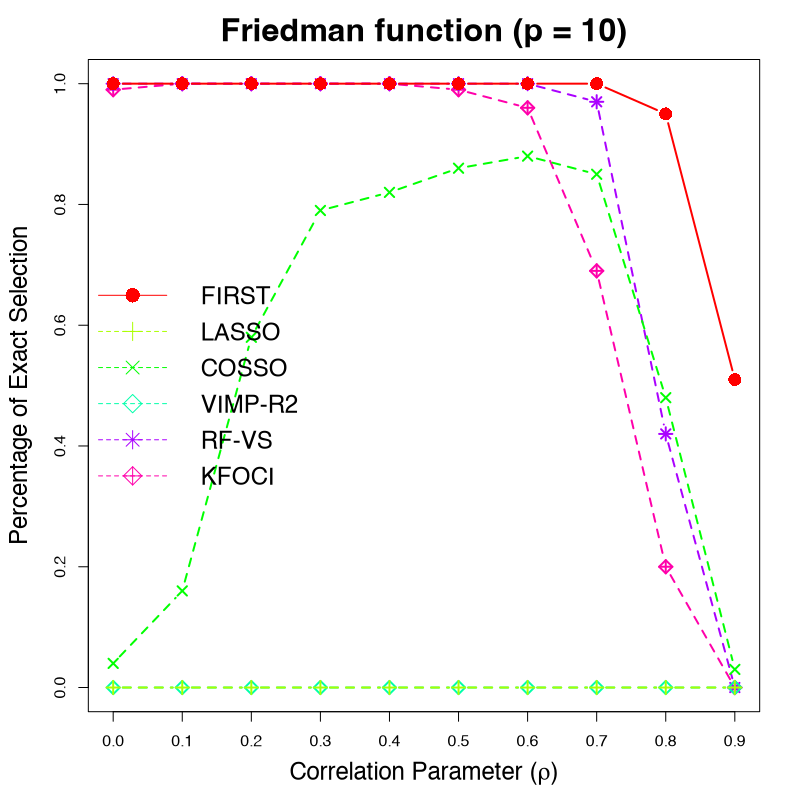}
        \caption{Friedman \eqref{eq:friedman}}
    \end{subfigure}
    \caption{Comparison of \texttt{FIRST} to other factor selection procedures on the Ishigami, Heavy-tailed, and Friedman function across different correlation levels. For each procedure, 100 independent runs are performed, and the percentage of times that the model variables are exactly selected is reported.}
    \label{fig:regression_selection}
\end{figure}

Now let us also examine the factor selection performance. Given that \texttt{FILTER}, \texttt{RF}, and \texttt{RF-SHAP} do not have a good criterion that automatically determines the number of variables to select, we omit them in this comparison. In addition to \texttt{LASSO}, \texttt{COSSO}, \texttt{VIMP-R2}, and \texttt{RF-VS}, we also compare \texttt{FIRST} to \texttt{KFOCI} (kernel feature ordering by conditional independence), a model-free forward stepwise variable selection procedure based on the kernel partial correlation \citep{huang2022kernel} implemented in the R package \texttt{KPC} \citep{huang2021kpc}. The default configurations suggested in the package are used for \texttt{KFOCI}. Figure~\ref{fig:regression_selection} reports the proportion of times that the true model variables are exactly selected, e.g., for Ishigami function \eqref{eq:ishigami} they are $\{X_{1},X_{2},X_{3}\}$. \texttt{FIRST} achieves the best possible performance on the Ishigami function, in particular, it selects $\{X_{1},X_{2},X_{3}\}$ exactly 100\% of the times for all correlation levels. For the Friedman function \eqref{eq:friedman}, \texttt{FIRST} remains the top performing method, especially when the inputs are strongly correlated ($\rho=0.9$): all the competitors fail completely while \texttt{FIRST} can still identify the true model variables exactly in about half of the times. For the Heavy-tailed nonlinear function \eqref{eq:heavy-tailed} though \texttt{FIRST} slightly falls short of \texttt{RF-VS} and \texttt{KFOCI}, its performance is still competitive, selecting the correct model variables at least 85\% of the cases. Overall \texttt{FIRST} stands out as the robust option against various levels of input correlation, and moreover it is computationally more efficient than its key competitors \texttt{COSSO}, \texttt{RF-VS} and \texttt{KFOCI}\footnote{For the Friedman function example with dimension $p=10$ and correlation $\rho=0.8$ in Figure~\ref{fig:friedman_08}, the average runtime of \texttt{KFOCI} is 70.15 seconds, which is 40 times greater than the runtime of \texttt{FIRST}.}.

\subsection{High Dimensional Setting}
\label{subsec:regression_high_dimensional_setting}

\begin{figure}[!t]
    \centering
    \begin{subfigure}{0.32\textwidth}
        \centering
        \includegraphics[width=\textwidth]{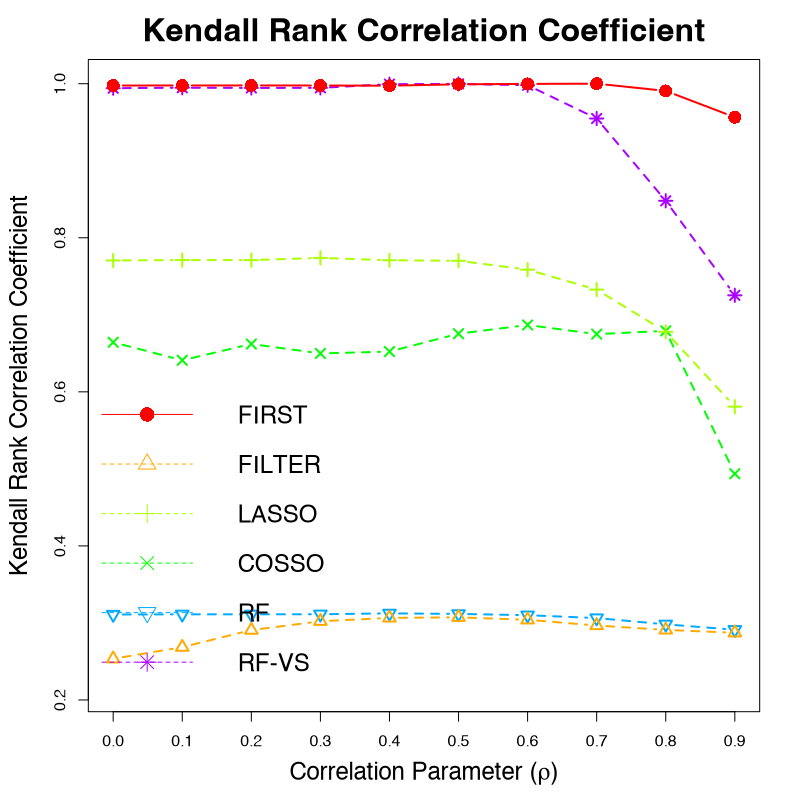}
    \end{subfigure}
    \begin{subfigure}{0.32\textwidth}
        \centering
        \includegraphics[width=\textwidth]{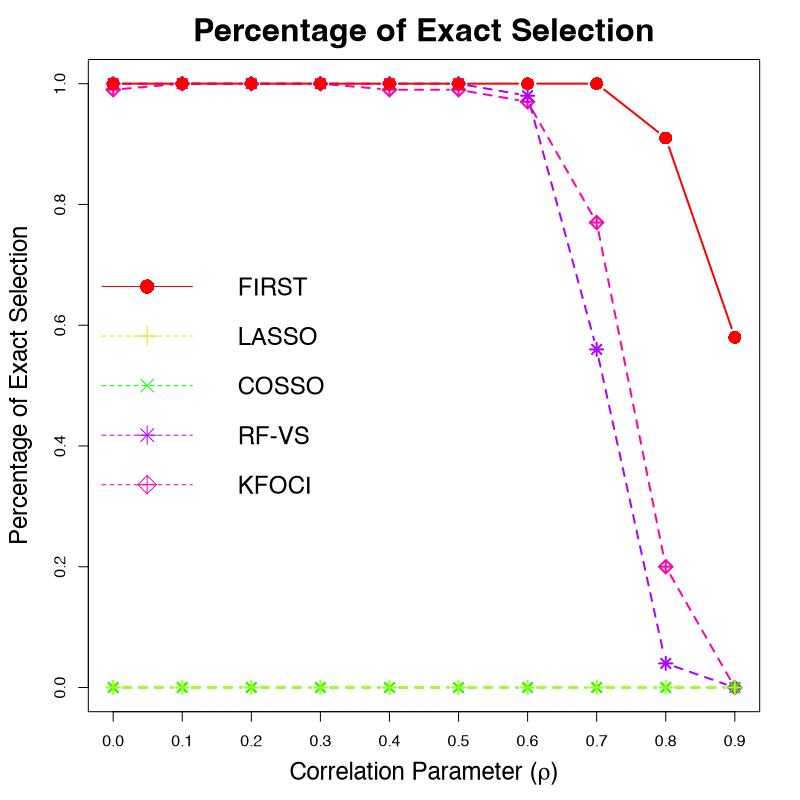}
    \end{subfigure}
    \begin{subfigure}{0.32\textwidth}
        \centering
        \includegraphics[width=\textwidth]{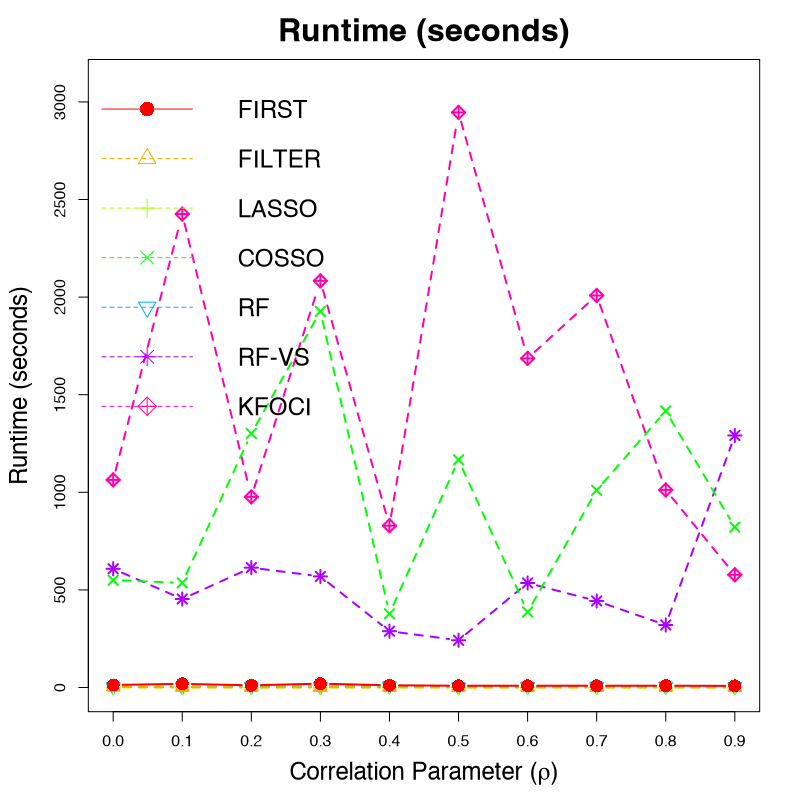}
    \end{subfigure}
    \caption{Comparison of \texttt{FIRST} to other factor importance / selection procedures on the $p=100$ dimensional Friedman \eqref{eq:friedman} function across different correlation levels. For each procedure, 100 independent runs are performed, and the following metrics are reported: average Kendall rank correlation coefficient \eqref{eq:kendall_tau} / the percentage of times that the model variables are exactly selected / the average runtime.}
    \label{fig:regression_hd}
\end{figure}

\begin{table}[t!]
    \centering
    \resizebox{\textwidth}{!}{
    \begin{tabular}{lcccccc}
        \toprule
        & \multicolumn{3}{c}{Ishigami \eqref{eq:ishigami}} & \multicolumn{3}{c}{Friedman \eqref{eq:friedman}} \\
        \cmidrule(lr){2-4}\cmidrule(lr){5-7}
        & $\rho=0.0$ & $\rho=0.5$ & $\rho=0.9$ & $\rho=0.0$ & $\rho=0.5$ & $\rho=0.9$  \\
        \midrule 
        $p=50$ & 1.00/1.00/6s & 0.99/1.00/6s & 1.00/1.00/6s & 0.99/1.00/10s & 1.00/1.00/9s & 0.96/0.67/5s \\
        $p=100$ & 1.00/1.00/7s & 1.00/1.00/7s & 1.00/1.00/7s & 1.00/1.00/13s & 1.00/1.00/9s & 0.96/0.58/8s \\
        $p=200$ & 1.00/1.00/24s & 1.00/1.00/24s & 1.00/1.00/24s & 1.00/1.00/22s & 1.00/1.00/25s & 0.96/0.62/21s \\
        $p=500$ & 1.00/1.00/59s & 1.00/1.00/58s & 1.00/1.00/58s & 1.00/0.99/55s & 1.00/1.00/103s & 0.95/0.60/84s \\
        $p=1000$ & 1.00/1.00/117s & 1.00/1.00/117s & 1.00/1.00/115s & 1.00/0.99/168s & 1.00/1.00/207s & 0.95/0.58/361s \\
        \bottomrule
    \end{tabular}
    }
    \caption{Performance of \texttt{FIRST} on the Ishigami and Friedman function with dimension $p=50,100,200,500,1000$. For each setting, 100 independent runs are performed, and the following metrics are reported: average Kendall rank correlation coefficient \eqref{eq:kendall_tau} / the percentage of times that the model variables are exactly selected / the average runtime.}
    \label{tab:regression_hd_first}
\end{table}

Figure~\ref{fig:regression_hd} compares the performance of \texttt{FIRST} to its competitors on the 100-dimensional Friedman function \eqref{eq:friedman}. Again only $X_{1},X_{7},X_{8},X_{9},X_{10}$ are important, and the rest 95 variables do not influence the output at all. We use $N=1{,}000$ noisy samples for the simulations. We can see that \texttt{FIRST} not only outperforms the other procedures on both factor importance and selection, but more importantly it is orders of magnitude faster over its key competitors \texttt{COSSO}, \texttt{RF-VS}, and \texttt{KFOCI}. Moreover, Table~\ref{tab:regression_hd_first} reports the simulation results of \texttt{FIRST} on $p=50,100,200,500,1000$. For the Ishigami \eqref{eq:ishigami} function, \texttt{FIRST} is able to (i) identify the true model parameters exactly and (ii) learn the importance ranking accurately on all dimensions. Similar findings are also observed for the Friedman \eqref{eq:friedman} function except on the extremely difficult strongly correlated case ($\rho=0.9$), yet the performance of \texttt{FIRST} does not degrade as the dimension increases. On the other hand, though the computational time of \texttt{FIRST} is reasonable for moderate ($p=200$) dimensional problem, it is getting computationally more expensive as the dimension further increases. To address this computational challenge, in Appendix~\ref{appendix:additional_simulation_results_on_regression} we provide a fast version of \texttt{FIRST} that borrows idea from the \emph{effect sparsity} principle \citep{wu2011experiments} to more efficiently filter out variables that are not important. Unfortunately it also comes with the trade-off of sacrificing some accuracy.  
 
\subsection{Real Data Example}
\label{subsec:regression_real_data_example}

In this subsection we study the performance of \texttt{FIRST} on some real world datasets. Let us start with the Abalone dataset \citep{nash1994population} that is available on the UCI Machine Learning repository \citep{asuncion2007uci}. This dataset consists of $N=4{,}177$ instances, and the objective is to predict the age of the abalone using the other eight physical measurements, including sex (male/female/infant), length (the longest shell measurement), diameter (perpendicular to length), height, whole weight, shucked weight (weight of the meat), viscera weight (weight of the gut after bleeding), and shell weight. Among these eight predictors, \texttt{FIRST} filters out height from the set of relevant variables for the age prediction. By fitting a random forest on the remaining seven variables, the out-of-bag prediction mean squared errors (OOB-MSE) is 4.59, achieving better results than the 4.61 OOB-MSE from the random forest fit on all eight predictors. Moreover, among the 256 possible subsets of the eight factors, the seven factors selected by \texttt{FIRST} yields the \emph{best random forest model} with the smallest OOB-MSE. Table~\ref{tab:abalone_importance} reports the factor importance estimated by \texttt{FIRST}, and we can see that the ranking of the importance measures is in good agreement with the predictive power proxy by the increment in OOB-MSE if the variable is dropped from the \emph{best random forest model}. This again demonstrate that the factor importance measure by \texttt{FIRST} nicely captures the predictive capability of the variables.  

\begin{table}[t!]
    \centering
    \resizebox{\textwidth}{!}{
    \begin{tabular}{lcc}
        \toprule
        Factor & \texttt{FIRST} Importance & Increment in OOB-MSE if Dropped \\
        \midrule 
        Sex (categorical) & 0.016 (6) & 0.098 (4) \\
        Length (continuous) & 0.012 (7) & 0.010 (7) \\
        Diameter (continuous) & 0.022 (4) & 0.020 (6) \\
        Height (continuous) & 0.000 (8) & 0.000 (8) \\
        Whole weight (continuous) & 0.040 (2) & 0.237 (2) \\
        Shucked weight (continuous) & 0.094 (1) & 0.597 (1) \\
        Viscera weight (continuous) & 0.019 (5) & 0.059 (5) \\
        Shell weight (continuous) & 0.031 (3) & 0.222 (3) \\
        \bottomrule
    \end{tabular}
    }
    \caption{The factor importance estimated by \texttt{FIRST} on the Abalone dataset. The increment in out-of-bag prediction mean squared errors (OOB-MSE) if the factor is dropped from the \emph{best random forest model} is also reported. The ranking is provided inside the parentheses.}
    \label{tab:abalone_importance}
\end{table}

Now we further study its factor selection performance on four other popular real data examples: Auto MPG data \citep{asuncion2007uci}, Boston Housing data \citep{harrison1978hedonic}, concrete compressive strength data \citep{yeh1998modeling}, and meat spectroscopy data \citep{thodberg1993ace}. For each dataset, we perform 80\%/20\% train/test split after filtering out the instances with missing values, next apply \texttt{FIRST} to select important factors from the train set, then fit a random forest on the train set using \texttt{FIRST} selected variables, and last evaluate it on the test set to obtain the test mean squared errors (MSE). We repeat this process for 100 times and record the 20\%, 50\%, and 80\% quantiles for (i) the number of selected factors and (ii) the ratio of test MSE to that from random forest fitting on the entire set of variables. From Table~\ref{tab:regression_real_comp} we can see that the test MSE ratio is generally around one, showing that the predictive power does not degrade much from using the variables selected by \texttt{FIRST}. Moreover, on the meat spectroscopy dataset, \texttt{FIRST} is able to reduce the problem dimension from 100 to around 4, while still maintaining a comparable prediction performance to the random forest fit with all 100 factors. 

\begin{table}[t!]
    \centering
    \resizebox{\textwidth}{!}{
    \begin{tabular}{lcc}
        \toprule
        Dataset & \# Selected Factors & Test MSE Ratio \\
        \midrule 
        Auto MPG ($N=392,\;p=7$) & 4[4,5] & 1.04[0.98,1.11] \\ 
        Boston Housing ($N=506,\; p=13$) & 8[6,9] & 1.02[0.93,1.14] \\
        Concrete Compressive Strength ($N=1030,\; p=8$) & 6[6,7] & 0.93[0.88,0.97] \\
        Meat Spectroscopy ($N=215,\;p=100$) & 4[2,6] & 0.93[0.77,1.07] \\
        \bottomrule
    \end{tabular}
    }
    \caption{The number of selected factors and the ratio between the test MSE obtained from the random forest fitted on the selected covariates and the entire set of covariates. The median [20\% quantile, 80\% quantile] are reported from 100 independent replications.}
    \label{tab:regression_real_comp}
\end{table} 

\section{Extension to Binary Classification Problems}
\label{sec:extension_to_binary_classification_problems}

In Section~\ref{sec:simulations_on_regression_problems} we have demonstrated the excellent performance of \texttt{FIRST} on the regression problems. A natural question is whether this can be extended for binary classification problems. To answer this question, we first present simulation results of \texttt{FIRST} (Algorithm~\ref{algo:first}) on both synthetic and real data binary classification problems, and then provide intuition on why \texttt{FIRST} can be applied directly without any modification. Different from regression that $N_{I}=2$ nearest-neighbor are sufficient for estimating the conditional variance, from our simulation we find that $N_{I}=3$ is more robust for the binary classification problem. For the results presented in this section, $N_{I}=3$ is used for \texttt{FIRST}.

\subsection{Synthetic Data Example}
\label{subsec:classification_synthetic_data_example}

Similar to the regression synthetic example, we model the input $X$ using Gaussian copula with marginals being uniform distributions and covariance matrix $\bm{\Sigma}$ defined by $\Sigma_{ij}=\rho^{|i-j|}$, where $0\leq\rho\leq 1$ controls the level of correlation. To obtain the binary outcome, we first convert the function output to a probability by the inverse-probit transformation, and then perform Bernoulli sampling using this probability, i.e., 
\begin{align*}
    Y = \text{Bernoulli}(\Phi(f^{*}(\bm{X}))),
\end{align*}
where $\Phi$ is the standard Gaussian cumulative distribution function. $N=1,000$ noisy binary output samples are generated for the simulations. 

\begin{figure}[!t]
    \centering
    \begin{subfigure}{\textwidth}
        \centering
        \includegraphics[width=0.32\linewidth]{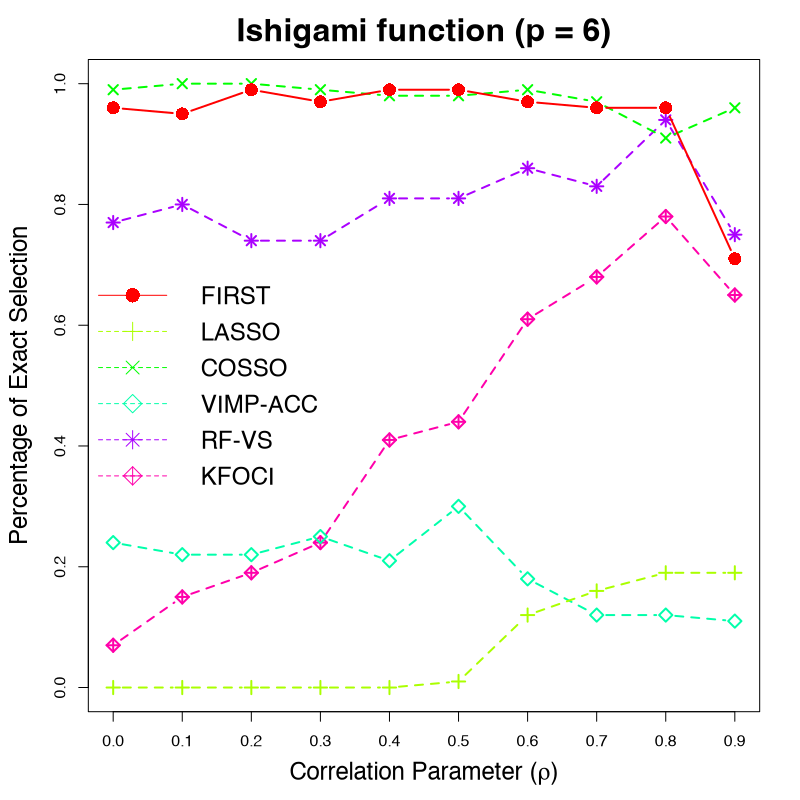}
        \hfill
        \includegraphics[width=0.32\linewidth]{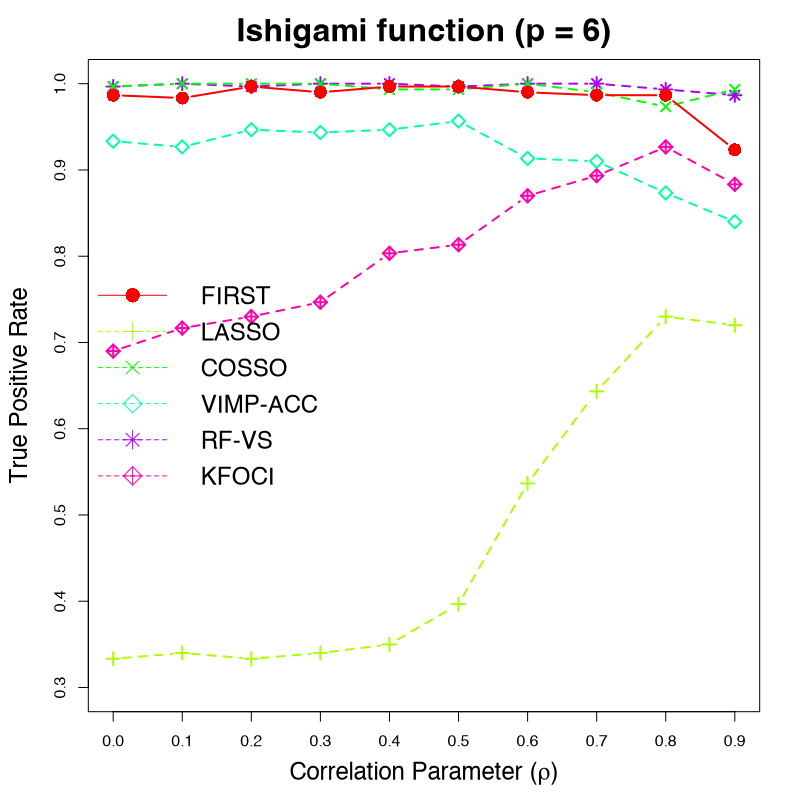}
        \hfill
        \includegraphics[width=0.32\linewidth]{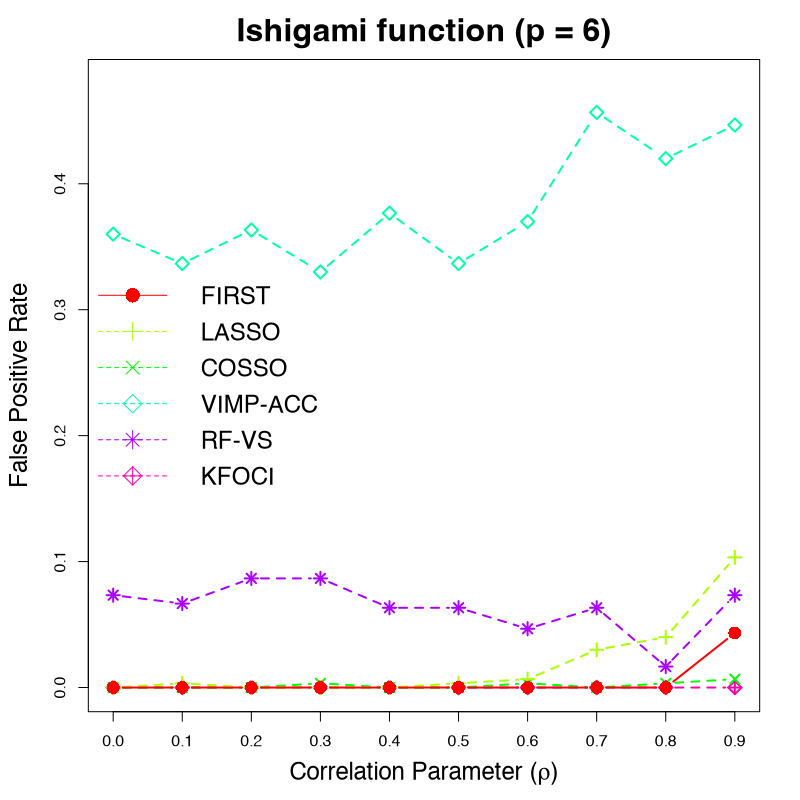}
        \caption{Ishigami \eqref{eq:ishigami}}
    \end{subfigure}
    
    \begin{subfigure}{\textwidth}
        \centering
        \includegraphics[width=0.32\linewidth]{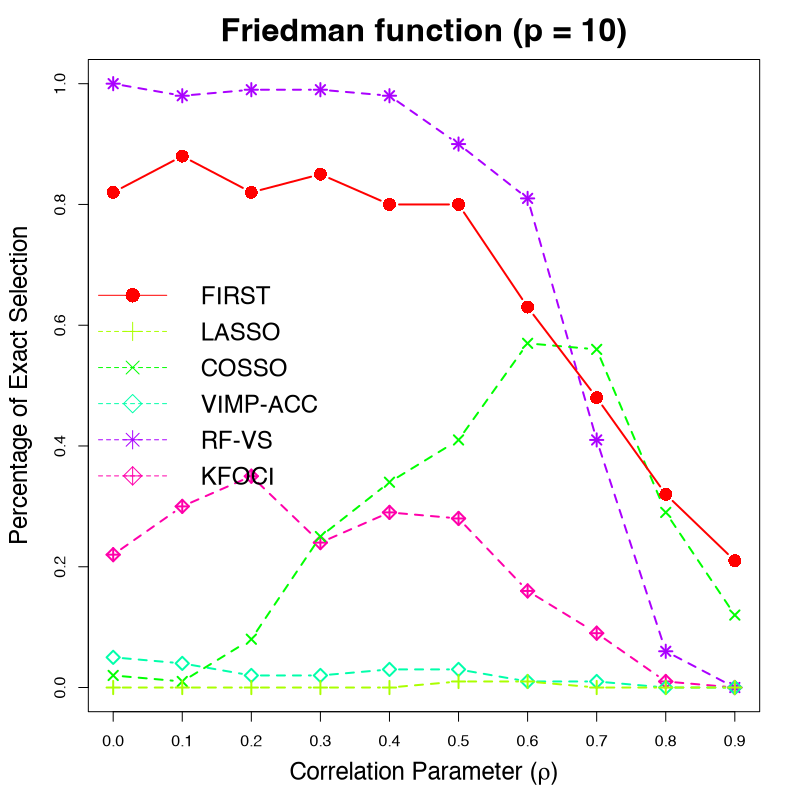}
        \hfill
        \includegraphics[width=0.32\linewidth]{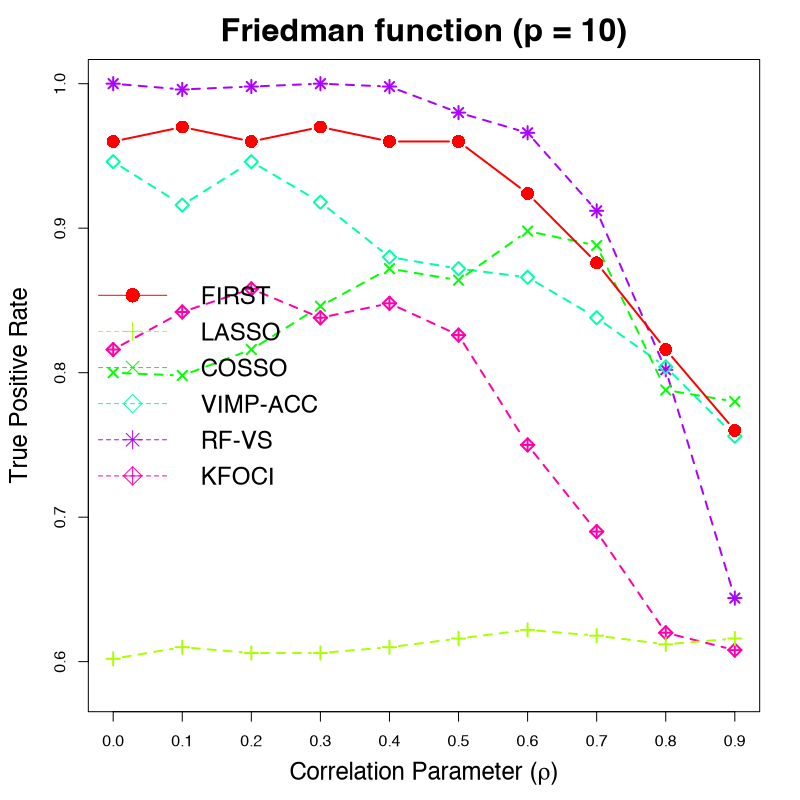}
        \hfill
        \includegraphics[width=0.32\linewidth]{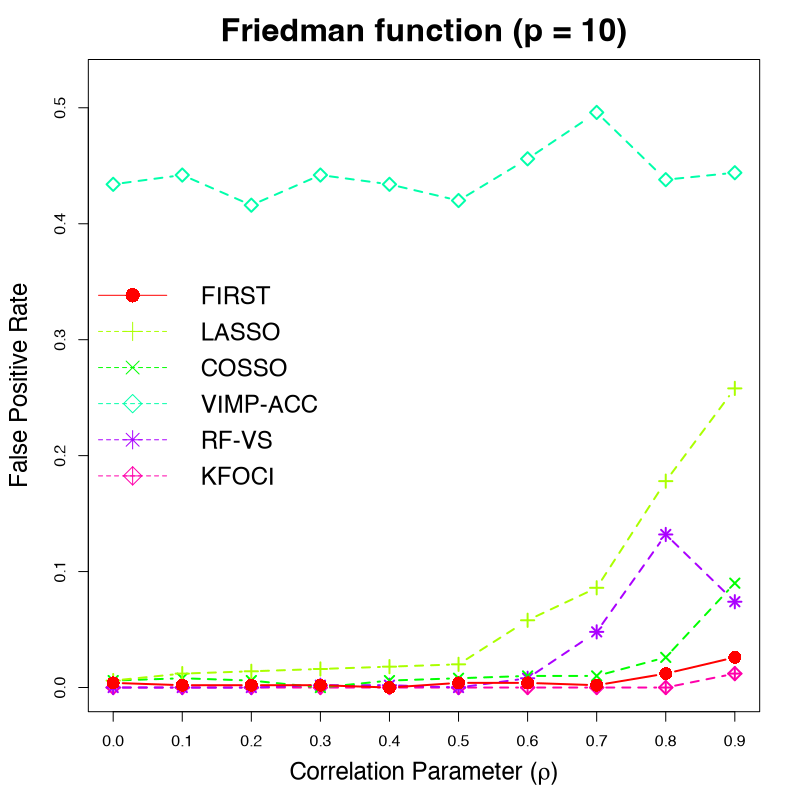}
        \caption{Friedman \eqref{eq:friedman}}
    \end{subfigure}
    
    \caption{Comparison of \texttt{FIRST} to other factor selection procedures on the Ishigami and the Friedman binary classification problem for different correlation ($\rho$) levels. For each procedure, 100 independent runs are performed, and the following metrics are reported: the percentage of times that the model variables are exactly selected / average true positive rate of the selected variables / average false positive rate of the selected variables.}
    \label{fig:classification_comp}
\end{figure}

For the binary classification problem, \citet{williamson2023general} suggest that (i) classification accuracy and (ii) area under the receiver operating characteristic curve (AUC) as the measures of predictiveness $V$ for computing the factor importance $\psi_{u}$ \eqref{eq:factor_importance}. The equivalence between the total Sobol' indices and the importance using $R^{2}$ predictiveness measure in the regression setting no longer holds here. Hence, there is no straightforward way to compute the groundtruth factor importance. We focus on the factor selection perspective for the performance evaluation. Three metrics are reported: (i) the proportion of times that the model variables are exactly selected, (ii) the true positive rate (TPR), the probability of success detection, of the selected variables, and (iii) the false positive rate (FPR), the probability of false alarm, of the selected variables. We compare \texttt{FIRST} to the following competitors: \texttt{LASSO}, \texttt{COSSO}, \texttt{RF-VS}, \texttt{VIMP-ACC}, and \texttt{KFOCI}. \texttt{VIMP-ACC} is the classification adaptation of the \texttt{VIMP-R2} considered in the regression comparison, which replaces the $R^{2}$ predictiveness measure by the classification accuracy\footnote{The simulations in \citet{williamson2023general} shows that using accuracy and AUC yield similar results.}. From Figure~\ref{fig:classification_comp} we can see that \texttt{FIRST} is among one of the top performers: \texttt{FIRST} is marginally less effective than \texttt{COSSO} on the Ishigami function \eqref{eq:ishigami}, and only outperformed by \texttt{RF-VS} on the Friedman function \eqref{eq:friedman}. \texttt{FIRST} achieves excellent false positive rate, demonstrating the reliability of the variables selected by \texttt{FIRST}. For the true positive rate, \texttt{FIRST} is not doing as good on the Friedman function, occasionally missing some true model variables. However, overall \texttt{FIRST} maintains a good balance between the true positive rate and the false positive rate compared to its competitors.  

\subsection{Real Data Example}
\label{subsec:classification_real_data_example}

\begin{table}[t!]
    \centering
    \resizebox{\textwidth}{!}{
    \begin{tabular}{lcc}
        \toprule
        Dataset & \# Selected Factors & Test Accuracy Ratio \\
        \midrule 
        A vs. S Letter Recognition ($N=1537,\;p=16$) & 8[7,9] & 0.99[0.99,1.00] \\
        Australian Credit Approval ($N=690,\;p=14$) & 5[4,6] & 0.98[0.97,1.01] \\
        Breast Cancer ($N=683,\;p=9$) & 5[4,6] & 0.99[0.98,1.00] \\
        Pima Indians Diabetes ($N=392,\;p=8$) & 3[2,4] & 0.98[0.94,1.03] \\
        Titanic ($N=714,\;p=7$) & 4[3,5] & 0.99[0.97,1.01] \\
        \bottomrule
    \end{tabular}
    }
    \caption{The number of selected factors and the ratio between the test accuracy obtained from the random forest fitted on the selected covariates and the entire set of covariates. The median [20\% quantile, 80\% quantile] are reported from 100 independent replications.}
    \label{tab:classification_real_comp}
\end{table}

Let us study the performance of \texttt{FIRST} on five real data applications: A vs. S Letter Recognition\citep{asuncion2007uci}, Australian credit approval \citep{asuncion2007uci}, Breast Cancer \citep{wolberg1990multisurface}, Pima Indians Diabetes \citep{ripley2007pattern}, and Titanic \citep{cukierski2012titanic}. Similar to the real data regression examples in Subsection~\ref{subsec:regression_real_data_example}, we first filter out the instances with missing values. Next, we split the dataset into 80\% training and 20\% testing and perform factor selection via \texttt{FIRST} on the train set. Last, we fit a random forest model on the train set using only the variables identified by \texttt{FIRST} and then evaluate the accuracy on the test set. Table~\ref{tab:classification_real_comp} reports the ratio of the aforementioned test accuracy to that obtained by fitting random forest on all the factors. We can see that with \texttt{FIRST}, using only about half of the factors is able to achieve a prediction performance that is comparable to the random forest fit with all the covariates. 

\subsection{Connection to the Impurity Measure}
\label{subsec:connection_to_the_impurity_measure}

We now provide some intuition on why $\texttt{FIRST}$ works for the binary classification problem. Recall that the total Sobol' index of $X_{i}$ is 
\begin{align*}
    S_{i}^{\text{tot}} = \frac{\mathbb{E}[\text{Var}\{Y|\bm{X}_{-i}\}]}{\text{Var}[Y]},
\end{align*}
where $Y=\mathbbm{1}(f^{*}(\bm{X})>0)$ is now a binary random variable. We have $\text{Var}[Y|\bm{X}_{-i}] = p_{-i}(1-p_{-i})$ where $p_{-i} = \mathbb{E}[\mathbbm{1}(f^{*}(\bm{X})>0)|\bm{X}_{-i}]$ is the probability of $Y=1$ conditional on $\bm{X}_{-i}$. In the binary classification problem, $p(1-p)$ is half of the Gini impurity measure, with maximum at $p=1/2$ and minimum at $p=0$ or $p=1$. When the full model $f^{*}(\bm{X})$ is known, the Gini impurity is 0 everywhere since we know whether $Y=1$ with probability one. The total Sobol' index $S_{i}^{\text{tot}}$ can then be interpreted as measuring the expected increment in the Gini impurity if $X_{i}$ is dropped from the model, and the larger the increment the more important $X_{i}$ is. This closely aligns with the importance definition of the tree models via the average reduction in the Gini impurity \citep{breiman2001random}. It is also straightforward to see that $X_{i}$ is not part of the model if and only if $S_{i}^{\text{tot}}=0$, and hence demonstrating that \texttt{FIRST} is appropriate for factor selection. 

\section{Conclusions}
\label{sec:conclusion}

In this article we propose \texttt{FIRST} (Algorithm~\ref{algo:first}), a novel procedure for estimating factor importance directly from noisy data. \texttt{FIRST} is motivated from the equivalence between the total Sobol' indices and the intrinsic population-level variable importance, circumventing the model fitting step that is required by the existing factor importance procedures in the literature. Extensive simulations are provided in Section~\ref{sec:simulations_on_regression_problems} and Section~\ref{sec:extension_to_binary_classification_problems} to demonstrate that \texttt{FIRST} not only provides trustworthy factor importance estimation and selection across problems with varying correlation levels and dimensions, but more importantly it outshines its main competitors by completing the task in only a fraction of the computational time.

\vspace{.25in}

\begin{center}
{\Large\bf Acknowledgments}
\end{center}

\noindent This research is supported by  U.S. National Science Foundation grants DMS-2310637 and DMREF-1921873.

\medskip

\bibliography{references}

\newpage

\setcounter{algorithm}{0}
\renewcommand{\thealgorithm}{S\arabic{algorithm}}%
\setcounter{table}{0}
\renewcommand{\thetable}{S\arabic{table}}%
\setcounter{figure}{0}
\renewcommand{\thefigure}{S\arabic{figure}}%
\setcounter{section}{0}

\section*{\LARGE{Appendices}}

\appendix

\section{Proof of Lemma~\ref{lm:expecected_mse}}
\label{appendix:proof_of_lemma}

\begin{lemma}
    \label{lm:expecected_mse}
    $\mathbb{E}[(\mathbb{E}[Y|X]-Y)^2] = \mathbb{E}[\textup{Var}(Y|X)]$.
\end{lemma}
\begin{proof}
\begin{align*}
    \mathbb{E}[(\mathbb{E}(Y|X)-Y)^2] &= \mathbb{E}[\mathbb{E}(Y|X)^2] + \mathbb{E}[Y^2] - 2\mathbb{E}[\mathbb{E}(Y|X)Y] \\
    &= \mathbb{E}[\mathbb{E}(Y|X)^2] + \mathbb{E}[Y^2] - 2\mathbb{E}[\mathbb{E}\{\mathbb{E}(Y|X)Y|X\}] \\
    &= \mathbb{E}[\mathbb{E}(Y|X)^2] + \mathbb{E}[Y^2] - 2\mathbb{E}[\mathbb{E}(Y|X)^2] \\
    &= \mathbb{E}[Y^2] - \mathbb{E}[\mathbb{E}(Y|X)^2] \\
    &= \left(\mathbb{E}[Y^2]-\mathbb{E}[Y]^2\right) - \left(\mathbb{E}[\mathbb{E}(Y|X)^2]-\mathbb{E}[Y]^2\right) \\
    &= \left(\mathbb{E}[Y^2]-\mathbb{E}[Y]^2\right) - \left(\mathbb{E}[\mathbb{E}(Y|X)^2]-\mathbb{E}[\mathbb{E}(Y|X)]^2\right) \\
    &= \text{Var}[Y] - \text{Var}[\mathbb{E}(Y|X)] \\
    &= \mathbb{E}[\text{Var}(Y|X)].
\end{align*}
\end{proof}

\section{Additional Simulation Results on Regression}
\label{appendix:additional_simulation_results_on_regression}

\begin{figure}[!t]
    \centering
    \begin{subfigure}{0.32\textwidth}
        \centering
        \includegraphics[width=\textwidth]{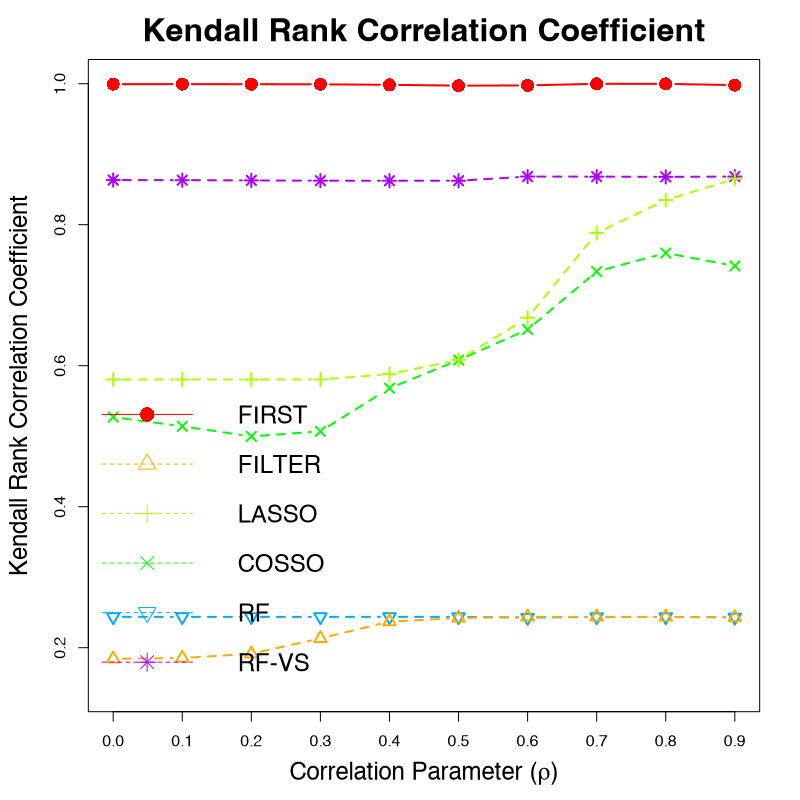}
    \end{subfigure}
    \begin{subfigure}{0.32\textwidth}
        \centering
        \includegraphics[width=\textwidth]{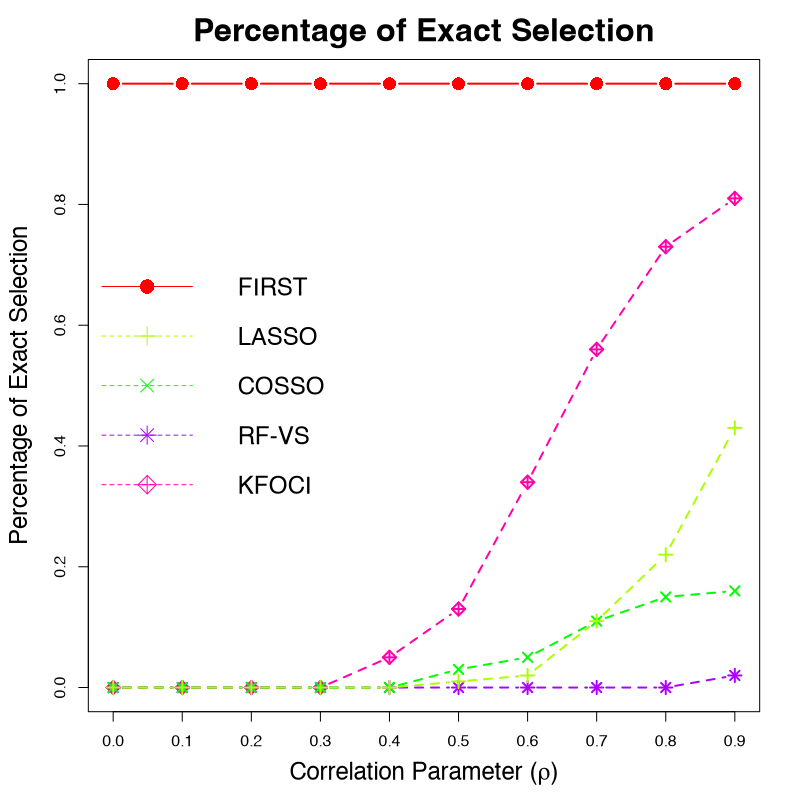}
    \end{subfigure}
    \begin{subfigure}{0.32\textwidth}
        \centering
        \includegraphics[width=\textwidth]{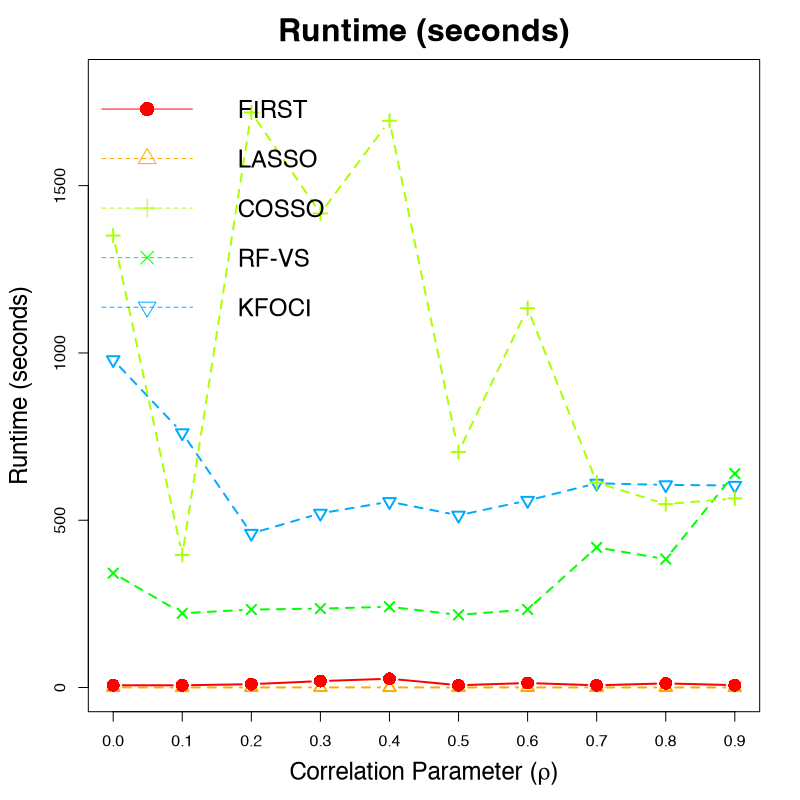}
    \end{subfigure}
    \caption{Comparison of \texttt{FIRST} to other factor importance / selection procedures on the $p=100$ dimensional Ishigami \eqref{eq:ishigami} function across different correlation levels. For each procedure, 100 independent runs are performed, and the following metrics are reported: average Kendall rank correlation coefficient \eqref{eq:kendall_tau} / the percentage of times that the model variables are exactly selected / the average runtime.}
    \label{fig:regression_hd_ishigami}
\end{figure}

Figure~\ref{fig:regression_hd_ishigami} compares the performance of \texttt{FIRST} to its competitors on the 100-dimensional Ishigami function \eqref{eq:ishigami} with $N=1{,}000$ noisy samples. Similar to the comparison on the Friedman function in Figure~\ref{fig:regression_hd}, \texttt{FIRST} outperforms the other methods for factor importance ranking in terms of the Kendall rank correlation coefficient. Moreover, \texttt{FIRST} is able to identify the true variables $\{X_{1},X_{2},X_{3}\}$ exactly 100\% of the times for different correlation levels. Compared to its key competitors \texttt{RF-VS} (for factor importance ranking) and \texttt{KFOCI} (for factor selection), \texttt{FIRST} only uses a fraction of their runtimes to achieve a better and more robust performance on this 100 dimensional example.

Though the computational time is fast for the 100 dimensional examples (Figure~\ref{fig:regression_hd} and ~\ref{fig:regression_hd_ishigami}), from Table~\ref{tab:regression_hd_first} we can see that \texttt{FIRST} is getting computationally more expensive as the dimension grows up. For the strongly correlated ($\rho=0.9$) 1{,}000-dimensional input Friedman function example, on average \texttt{FIRST} takes 361 seconds to complete the computation. To further speed up \texttt{FIRST}, we incorporate the \emph{effect sparsity} principle\citep{wu2011experiments}. The key is to filter out irrelevant factors more efficiently to obtain a parsimonious model. Instead of considering all factors except for the selected ones as candidate in each forward selection step, we only keep the variables that are likely to improve the output variance that can be explained in the candidate set, i.e., in each forward selection step we remove variables that would diminish the explainable variance from the candidate set. Algorithm~\ref{algo:first_fast} details this fast version of \texttt{FIRST}, and the difference from \texttt{FIRST} (Algorithm~\ref{algo:first}) is highlighted in red. Table~\ref{tab:regression_hd_first_fast} reports the comparison of \texttt{FIRST} and \texttt{FIRST-FAST} on the $p=50,100,200,500,1000$ dimensional Friedman function. Clearly we can see that \texttt{FIRST-FAST} is able to reduce the runtime of \texttt{FIRST} by more than 3 times, but in the sacrifice of worse true positive rate, i.e., \texttt{FIRST-FAST} is more likely to leave out some crucial variables. If the computational time is not a major bottleneck, we would recommend \texttt{FIRST} for its better accuracy. 

\begin{algorithm}[t!]
    \caption{\texttt{FIRST-FAST}.}
    \label{algo:first_fast} 
    \begin{algorithmic}[1]
        \State \textbf{Input:} (i) noisy samples $\{(\bm{x}^{(n)}\in\mathbb{R}^{p},y^{(n)}\in\mathbb{R})\}_{n=1}^{N}$, (ii) number of samples for inner loop Monte Carlo $N_{I}$, and (iii) number of samples for outer loop Monte Carlo $N_{O}$.
        \State \textbf{Initialization:} $k=-1$, $A \gets \emptyset$, $\{i_{0}\}=\emptyset$, $\hat{V}_{\emptyset}=0$, \textcolor{red}{$C \gets \{1\!:\!p\}$};
        \Do 
            \State $k \gets k + 1$;
            \State $A \gets A \cup \{i_{k}\}$;
            \State \textcolor{red}{$C \gets C\backslash \{i_{k}\}$};
            \State Choose \textcolor{red}{$i_{k+1}\in C$} such that $\hat{V}_{A\cup\{i_{k+1}\}}$ \eqref{eq:consistent_vare} is maximized, i.e., 
            \begin{align*}
                i_{k+1} = \arg\max_{i\in 1:p\backslash A}\hat{V}_{A\cup\{i\}},
            \end{align*}
            \State \textcolor{red}{and remove $i\in C$ from $C$ such that $\hat{V}_{A\cup\{i\}} < \hat{V}_{A}$, i.e., 
            \begin{align*}
                C \gets C\backslash \{i\in C: \hat{V}_{A\cup\{i\}} < \hat{V}_{A}\};
            \end{align*}}
        \doWhile{\textcolor{red}{$|C| > 0$}};
        \State $\{\hat{S}_{i}^{\text{tot}}\}_{i\in A} \gets \texttt{NANNE-BE}(\{(x_{A}^{(n)},y^{(n)})\}_{n=1}^{N},N_{I},N_{O})$; 
        \State $\hat{\psi}_{i} \gets \hat{S}_{i}^{\text{tot}}$ for $i\in A$ and $\hat{\psi}_{i} \gets 0$ for $i\in 1\!:\!p\backslash A$;
        \State \textbf{Output:} the factor importance $\{\hat{\psi}_{i}\}_{i=1}^{p}$.
    \end{algorithmic}
\end{algorithm}

\begin{table}[t!]
    \centering
    \resizebox{\textwidth}{!}{
    \begin{tabular}{lcccccc}
        \toprule
        & \multicolumn{2}{c}{$\rho=0.0$} & \multicolumn{2}{c}{$\rho=0.5$} & \multicolumn{2}{c}{$\rho=0.9$} \\
        \cmidrule(lr){2-3}\cmidrule(lr){4-5}\cmidrule(lr){6-7}
        & \texttt{FIRST} & \texttt{FIRST-FAST} & \texttt{FIRST} & \texttt{FIRST-FAST} & \texttt{FIRST} & \texttt{FIRST-FAST} \\
        \midrule 
        $p=50$  & 0.99/1.00/0.00/10s & 0.88/0.80/0.00/3s & 1.00/1.00/0.00/9s & 0.93/0.87/0.00/3s & 0.96/0.94/0.00/5s & 0.91/0.84/0.00/2s \\
        $p=100$ & 1.00/1.00/0.00/13s & 0.89/0.81/0.00/9s & 1.00/1.00/0.00/9s & 0.94/0.89/0.00/3s & 0.96/0.93/0.00/8s & 0.92/0.85/0.00/3s \\
        $p=200$ & 1.00/1.00/0.00/22s & 0.89/0.81/0.00/6s & 1.00/1.00/0.00/25s & 0.94/0.89/0.00/7s & 0.96/0.93/0.00/21s & 0.92/0.85/0.00/6s \\
        $p=500$ & 1.00/1.00/0.00/55s & 0.89/0.81/0.00/14s & 1.00/1.00/0.00/103s & 0.94/0.88/0.00/28s & 0.95/0.92/0.00/84s & 0.91/0.84/0.00/24s \\
        $p=1000$ & 1.00/1.00/0.00/168s & 0.88/0.81/0.00/42s & 1.00/1.00/0.00/207s & 0.94/0.88/0.00/55s & 0.95/0.92/0.00/361s & 0.92/0.85/0.00/105s \\
        \bottomrule
    \end{tabular}
    }
    \caption{Comparison of \texttt{FIRST} and \texttt{FIRST-FAST} on the Friedman function with dimension $p=50,100,200,500,1000$. For each setting, 100 independent runs are performed, and the following metrics are reported: average Kendall rank correlation coefficient \eqref{eq:kendall_tau} / true positive rate / false positive rate / the average runtime.}
    \label{tab:regression_hd_first_fast}
\end{table}

\section{Minor Implementation Details}
\label{appendix:minor_implementation_details}

We want to point out a minor implementation detail of the \texttt{NANNE} estimator (Algorithm~\ref{algo:nanne}) that is different from the nearest-neighbor estimator of \citet{broto2020variance}. For the nearest-neighbor estimator, when there are multiple nearest neighbors that are at the same distance, \citet{broto2020variance} suggests to break the tie by random selection and obtain exactly $N_{I}$ nearest neighbors for the conditional variance computation. In actual implementation, we notice that the random selection introduces noise to the estimation, especially when majority of the variables are categorical. Hence, instead of breaking tie by random selection, our proposed \texttt{NANNE} estimator computes the conditional variance using all the samples that are within the distance of the $N_{I}$-th nearest neighbor. This minor modification exhibits a more robust performance on the real world datasets.      

\end{document}